\documentclass[11pt]{article}
\usepackage{amsthm,amssymb,amsfonts,amsmath,amscd}
\usepackage{nccmath}
\usepackage{graphicx}
\usepackage[pdfa]{hyperref}

\makeatletter \@addtoreset{equation}{section} \makeatother

\newtheorem{theorem}{Theorem}[section]
\newtheorem{lemma}{Lemma}[section]

\newtheorem{proposition}{Proposition}[section]

\newcommand{\mdet}{\mathrm{det}}
\newcommand{\intd}{\displaystyle\int}
\newcommand{\Tr}{\mathrm{Tr}\,}
\newcommand{\Str}{\mathrm{Str}\,}
\newcommand{\Sdet}{\mathrm{Sdet}\,}

\begin{document}

\title{Finite-rank complex deformations of random band matrices: sigma-model approximation}
\author{ Mariya Shcherbina
\thanks{Institute for Low Temperature Physics, Kharkiv, Ukraine, e-mail: shcherbi@ilt.kharkov.ua} \and
 Tatyana Shcherbina
\thanks{ Department of Mathematics, University of Wisconsin - Madison, USA, e-mail: tshcherbyna@wisc.edu. }}
\date{}
\maketitle

\begin{abstract}
We study the distribution of complex eigenvalues $z_1,\ldots, z_N$ of random Hermitian $N\times N$ block band matrices with a complex deformation of a finite rank. Assuming that the width of the band $W$ grows faster than $\sqrt{N}$, we proved that the limiting density of $\Im z_1,\ldots, \Im z_N$ in a sigma-model approximation coincides with that for the Gaussian Unitary Ensemble.   The method follows the techniques of \cite{SS:sigma}.

\end{abstract}

\section{Introduction}\label{s:1}

The complex eigenvalues of non-Hermitian random matrices have attracted
much research interest due to their relevance to several branches of theoretical physics, and in particular to the study of scattering chaotic systems.
According to the works  \cite{Ver:85}, \cite{Sok:89}, universal properties of the poles of the scattering matrix $S(E)$ in the complex plane
can be modelled by $N$ complex eigenvalues $z_n$, $\Im z_n\le 0$ of so-called ``effective non-Hermitian Hamiltonian"
\begin{align}\label{H_eff}
\mathcal{H}_{eff}=H-i\Gamma, 
\end{align}
where $H$ is a random matrix ensemble with an appropriate symmetry (e.g., Hermitian or real symmetric), and $\Gamma$ is a positive deformation
of a rank $M\ll N$. More details of the approach can be found, e.g., in reviews \cite{FyoSav:11}, \cite{MRW:10},
\cite{FyoSom:03} and references therein.

One of the most interesting questions about the spectral statistics of $H_{eff}$ is the distribution of $\Im z_i$ (i.e. ``resonance widths").
In contrast to the classical non-Hermitian models such as Ginibre ensemble (random matrices with iid entries), if $M$ is fixed and $N\to\infty$, 
matrices $\mathcal{H}_{eff}$ are weakly non-Hermitian, and so $\Im z_i$ are of order of the typical spacing $\omega$ between eigenvalues of $H$,
i.e. $O(1/N)$. It is also expected that the spectral
fluctuations on the $\omega$-scale  is universal, i.e. independent of the particular form of the
distribution of $H$ or the energy dependence of $\omega$.

For the case $H$ taken from Gaussian Unitary Ensemble (GUE) the probability density of the scaled $\Im z_i$ was obtained in 
\cite{FyoSom:96}, \cite{FyoKh:99} for any finite $M$ (for some related models see review \cite{FyoSom:03} and references therein).
Let us mention also that the cases of non-Hermitian symmetry, and in particular real symmetric case, are much more involved,
and is not well-enough studied even for $H$ taken from Gaussian Orthogonal Ensemble (there are only some partial results
for $M=1$, see physical papers \cite{SFyT:99},\cite{FyoOs:21} for GOE; let us also mention the paper \cite{Koz:17} that gives
joint probability distribution of $z_i$ for rank-one perturbation of general $\beta$-ensembles).

In this paper we consider $H$ to be a one-dimensional Hermitian block band matrix (block RBM). The 1d block RBM are the special class of Wegner's orbital models (see \cite{We:79}), i.e.
Hermitian $N\times N$ matrices $H_N$  with complex zero-mean random Gaussian entries $H_{jk,\alpha\beta}$,
where $j,k =1,..,n$ (they parametrize the lattice sites) and $\alpha, \gamma= 1,\ldots, W$ (they
parametrize the orbitals on each site), $N=nW$, such that
\begin{equation}\label{H}
\langle H_{j_1k_1,\alpha_1\gamma_1}H_{j_2k_2,\alpha_2\gamma_2}\rangle=\delta_{j_1k_2}\delta_{j_2k_1}
\delta_{\alpha_1\gamma_2}\delta_{\gamma_1\alpha_2} J_{j_1k_1}
\end{equation}
with
\begin{equation}\label{J_old}
J=1/W+\tilde\beta\Delta/W,
\end{equation}
where $W\gg 1$ and $\Delta$ is the discrete Laplacian on $\{1,2,..,n\}$. 
The probability law of $H_N$ can be written in the form
\begin{equation}\label{pr_l}
P_N(d H_N)=\exp\Big\{-\dfrac{1}{2}\sum\limits_{j,k=1}^n\sum\limits_{\alpha,\gamma=1}^W
\dfrac{|H_{jk,\alpha\gamma}|^2}{J_{jk}}\Big\}dH_N.
\end{equation}
The density of states $\rho$ of a general class of RBM with $W\gg 1$ is given by the well-known Wigner semicircle law (see
\cite{BMP:91, MPK:92}):
\begin{equation}\label{rho_sc}
\rho(E)=(2\pi)^{-1}\sqrt{4-E^2},\quad E\in[-2,2].
\end{equation}
The main feature of RBM is that their local spectral statistics is conjectured to exhibit  the crossover at $W=\sqrt{N}$: for $W\gg \sqrt{N}$ the eigenvectors are expected to be
delocalized and the local spectral statistics is governed by the Wigner-Dyson (GUE/GOE) statistics, and for $W\ll \sqrt{N}$ the eigenvectors  are localized
and the local spectral statistics is Poisson. The conjecture is supported by the physical derivation due to Fyodorov and Mirlin (see \cite{FM:91}) based on supersymmetric formalism, but is not proved in the full generality yet. For the general  RBM the delocalization is proved for $W\gg N^{3/4}$ (see the review \cite{B:rev} and references 
therein). For  the more specific Gaussian model (\ref{H}) -- (\ref{J_old}),  the Wigner-Dyson local statistics is proved up to the optimal regime $W\gg \sqrt{N}$
first in the so-called sigma-model approximation \cite{SS:sigma}, and then in the full model  \cite{SS:Un} by the application of the supersymmetric transfer matrix approach.

The main advantage of the SUSY techniques here is that the main spectral
 characteristics of the model (\ref{H}) -- (\ref{J_old}) such as a density of states, spectral correlation functions, $\mathbb{E}\{|G_{jk}(E+i\varepsilon)|^2\}$, etc. can
 be expressed via SUSY as the averages of certain observables in nearest-neighbour statistical mechanics models on a box in $\mathbb{Z}$, which allows to 
 combine the SUSY techniques with a transfer matrix approach. However,  the rigorous application of the techniques to the main spectral characteristics of RBM is quite difficult due to
the complicated structure and non self-adjointness of the corresponding transfer operator. So it is easier to apply it first to the so-called sigma-model approximation, 
which is often used by physicists to study complicated statistical mechanics systems.  In such approximation spins of the statistical mechanics model take values in
some symmetric space ($\pm 1$ for Ising model, $S^1$ for the rotator, $S^2$ for the classical
Heisenberg model, etc.). It is expected that sigma-models have all the qualitative physics
of more complicated models with the same symmetry.  The sigma-model approximation for RBM was introduced by Efetov (see \cite{Ef}),
and the spins there are $4\times 4$ matrices with both complex and Grassmann entries (this approximation was studied in \cite{FM:91}, \cite{FM:94}).
The rigorous application of the techniques to the correlation functions of  (\ref{H}) -- (\ref{J_old}) was developed in \cite{SS:sigma}.

The aim of the current paper is to derive the sigma-model approximation for the limiting density of the imaginary parts of the eigenvalues of $H_{eff}$ of (\ref{H_eff}) with
$H$ of (\ref{H}), and, following the techniques of \cite{SS:sigma}, prove that its limiting behaviour in the delocalized regime $W\gg \sqrt{N}$ coincides with that for $H=\hbox{GUE}$.

Define
\begin{equation}\label{cal_H}
\mathcal{H}=H_N+i\Gamma_M,
\end{equation}
with $H_N$ of (\ref{H}) -- (\ref{J_old}), where $\Gamma_M$ is a $N\times N$ matrix 
\begin{equation}\label{Gamma}
\Gamma_M=\left(\begin{matrix}
\gamma_1&0&\ldots &\ldots &0 &\ldots &0 \\
0&\gamma_2&0&\ldots&0 &\ldots &0\\
\vdots&\vdots&\ddots&\vdots&\vdots &\vdots &\vdots \\
0&\ldots&0&\gamma_M&0&\ldots&0\\
0&\ldots&\ldots&0&0&\ldots&0\\
\vdots& \vdots&\vdots&\vdots&\vdots&\vdots&\vdots\\
0&\ldots&\ldots &0 &0 &\ldots &0
\end{matrix}\right)
\end{equation}
with some fixed $\gamma_i>0$ and fixed $M$. Notice that for convenience we have changed the sign of $\Gamma_M$ in order to get positive $\Im z_i$.

In order to access the  density $\rho(x,y)$ of complex eigenvalues $z_j=x_j+iy_j$ one can use the formula (see \cite{FyoSom:96} and reference therein)
\[
\rho_N(x,y)=-\dfrac{1}{4\pi N}\lim\limits_{\kappa\to 0}\partial^2\Phi(x,y,\kappa)
\]
with
\[
\Phi(x,y,\kappa)=-\dfrac{1}{N}\log \mdet\Big((\mathcal{H}-x-iy)(\mathcal{H}-x-iy)^*+\dfrac{\kappa^2}{N^2}\Big)
\]
where $\partial^2$ stands for the two-dimensional Laplacian and  a positive parameter 
  $\kappa$is added to regularize the logarithm. 
  
 Introduce the generating function
\begin{equation}\label{Z}
Z_{\beta nW}(\kappa,z_1,z_2)=\mathbb{E}\left[\dfrac{\mdet\Big\{(\mathcal{H}-z_1)(\mathcal{H}-z_1)^*+\dfrac{\kappa^2}{N^2}\Big\}}{\mdet\Big\{(\mathcal{H}-z_2)(\mathcal{H}-z_2)^*+\dfrac{\kappa^2}{N^2}\Big\}}\right],
\end{equation}
where $z_1$ and $z_2$ are auxiliary spectral parameters in the vicinity of $E + iy/N$:
\begin{equation}\label{small_z}
z_l=E_l+\dfrac{iy_l}{N},\quad E_l=E+\dfrac{x_l}{N},\quad l=1,2.
\end{equation}
Given $Z_{\beta nW}$, the density can be obtained using the following
identity (see \cite{FyoSom:96} and references therein):
\begin{equation*}
\rho_N(E,y/N)=\dfrac{1}{4\pi}\lim\limits_{\kappa\to 0}\lim\left(\dfrac{\partial}{\partial y_1}\Big(\lim\limits_{y_2\to y_1}\dfrac{\partial Z_{n,W}}{\partial y_2}\Big)
+\dfrac{\partial}{\partial x_1}\Big(\lim\limits_{x_2\to x_1}\dfrac{\partial Z_{n,W}}{\partial x_2}\Big)\right)\Bigg|_{\begin{array}{l}y_1=y,\\ x_1=0 \end{array}}
\end{equation*}
Following \cite{SS:sigma}, to derive sigma-model approximation of $Z_{\beta nW}$ for the model(\ref{H}) -- (\ref{J_old}), we take $\beta$ in (\ref{J_old}) of order $1/W$, i.e. put
\begin{equation}\label{J}
J=1/W+\beta\Delta/W^2.
\end{equation}
The first main result states that in the model (\ref{J}) with fixed $\beta$ and $n$, and with $W\to\infty$, the function
$Z_{\beta nW}(\kappa,z_1,z_2)$ of (\ref{Z}) converges to the value given by the sigma-model approximation. More precisely, we get
\begin{theorem}\label{thm:sigma_mod} 
Given $Z_{\beta nW}(\kappa,z_1,z_2)$ of (\ref{Z}), (\ref{cal_H}), (\ref{J}), and  (\ref{Gamma}), any fixed $\beta$, $n$, $\kappa>0$, $z_1$, $z_2$ of (\ref{small_z}), and $|E|\le\sqrt2$, we have,   as $W\to\infty$:
\begin{align*}
Z_{\beta nW}(\kappa,z_1,z_2)\to Z_{\beta n}(\kappa,z_1,z_2),
\end{align*}
where
\begin{align}\label{sigma-mod}
Z_{\beta n}(\kappa,z_1,z_2)=e^{E(x_1-x_2)}\int &\exp\Big\{-\dfrac{\tilde\beta}{4}\sum\Str Q_jQ_{j-1}+\dfrac{c_0}{2n}\sum \Str Q_j\Lambda_{\kappa, y_1,y_2}\Big\}\\ 
&\times\prod\limits_{a=1}^M\Sdet^{-1}\Big(Q_1-\dfrac{iE}{2\pi\rho(E)}+\dfrac{i\gamma_a}{\pi\rho(E)}\mathcal{L}\Sigma\Big)d Q, \notag
\end{align}
  $\tilde\beta=(2\pi\rho(E))^2\beta$, $c_0=2\pi\rho(E)$ with $\rho$ of (\ref{rho_sc}), $U_j\in \mathring{U}(2)$, $S_j\in \mathring{U}(1,1)$ (see notation (\ref{U}) below),
and $Q_j$ are $4\times 4$ supermatrices with commuting
diagonal and anticommution off-diagonal $2\times 2$ blocks
\begin{align*}
Q_j=\left(\begin{array}{cc}
U_j^*&0\\
0&S_j^{-1}
\end{array}\right)\left(\begin{array}{cc}
(I+2\hat\rho_j\hat\tau_j)L&2\hat\tau_j\\
2\hat\rho_j&-(I-2\hat\rho_j\hat\tau_j)L
\end{array}\right)\left(\begin{array}{cc}
U_j&0\\
0&S_j
\end{array}\right),
\end{align*}
\begin{align*}
d Q=\prod d Q_j,\quad d Q_j=(1-2n_{j,1}n_{j,2})\, d \rho_{j,1}d \tau_{j,1}\,d \rho_{j,2}d \tau_{j,2}\, d U_j\,d S_j
\end{align*}
with
\begin{align*}
&n_{j,1}=\rho_{j,1}\tau_{j,1},
\quad n_{j,2}=\rho_{j,2}\tau_{j,2},\\ \notag 
&\hat\rho_j=\mathrm{diag}\{\rho_{j1},\rho_{j2}\},\quad \hat\tau_j=\mathrm{diag}\{\tau_{j1},\rho_{j2}\},\quad L=\mathrm{diag}\{1,-1\}
\end{align*} 
Here $\rho_{j,l}$, $\tau_{j,l}$, $l=1,2$ are anticommuting Grassmann variables,
\begin{align}\label{sdet_def}
\Str \left(\begin{array}{cc}
A&\sigma\\
\eta&B
\end{array}\right)=\Tr B-\Tr A,\quad \Sdet\left(\begin{array}{cc}
A&\sigma\\
\eta&B
\end{array}\right)=\dfrac{\mdet(B- \eta A^{-1}\sigma)}{\mdet A},
\end{align}
and 
\begin{align*}
\Lambda_{\kappa,y_1,y_2}=\left(\begin{matrix}
\kappa&-iy_1&0&0\\
iy_1&-\kappa&0&0\\
0&0&\kappa&-iy_2\\
0&0&iy_2&-\kappa
\end{matrix}\right),\quad \mathcal{L}=\left(\begin{matrix}
I_2&0\\
0&-I_2
\end{matrix}\right),\quad \Sigma=\left(\begin{matrix}
\sigma&0\\
0&\sigma
\end{matrix}\right).
\end{align*}\end{theorem}
The next theorem gives asymptotic behaviour of $Z_{\beta n}$ in the delocalized regime $\beta\gg n$ as $n,\beta\to \infty$:
\begin{theorem}\label{thm:2} Given $Z_{\beta n}(\kappa,z_1,z_2)$ of (\ref{sigma-mod}), we have in the limit $\beta\to\infty, n\to\infty$
with \\ $\beta>n\log^3n$
\begin{align}\label{t2.1}
Z_{\beta n}(\kappa,z_1,z_2)\to e^{E(x_1-x_2)}\int\dfrac{\exp\big\{\pi\rho(E) \Str Q\Lambda_{\kappa, y}\big\}}{\prod_{a=1}^M\Sdet
\Big(Q-\frac{iE}{2\pi\rho(E)}+\frac{i\gamma_a}{\pi\rho(E)}\mathcal{L}\Sigma\Big) }dQ,
\end{align}
which coincides with $Z(\kappa,z_1,z_2)$ for the GUE. Therefore,  the limiting distribution of the imaginary parts of the eigenvalues of $\mathcal{H}$ of (\ref{cal_H})
with $H_N$ of (\ref{H}) -- (\ref{J_old}) in the sigma-model approximation coincides with that for $H_N=GUE$ obtained in \cite{FyoSom:96}. 
\end{theorem}
The paper is organized as follows. We are going to give a detailed proof for the case $M=1$ and explain some minor correction that should be done to prove the general case. In Section 2 we obtain the SUSY integral representation of $Z_{\beta nW}$ of (\ref{Z}). Section 3 is devoted to the derivation of sigma-model approximation, i.e. to the proof of Theorem \ref{thm:sigma_mod}. In Section 4 we prove Theorem \ref{thm:2} relying on the similar study in \cite{SS:sigma}.

\subsection{Notation}
We denote by $C$, $C_1$, etc. various $n$, $\beta$, $W$-independent quantities below, which
can be different in different formulas. Integrals
without limits denote the integration (or the multiple integration) over the whole
real axis, or over the Grassmann variables. 

Moreover,
\begin{itemize}

    \item[$\bullet$] $N=Wn;$


    \item[$\bullet$] indices $i,j,k$ vary from 1 to $n$ and correspond to the number of  block in $H_N$, index $l$ is always $1$ or $2$ (this is the field index),
    and Greek indices $\alpha, \gamma$ vary from $1$ to $W$ and correspond to the position of
    the element in the block;


    \item[$\bullet$] variables $\phi$ and $\Phi$ with different indices are complex variables or vectors
    correspondingly; if $x_j$ means some variable (vector or matrix) which corresponds to the site $j=1,..,n$, then $x$ means
    vector $\{x_j\}_{j=1}^n$,  $dx=\prod dx_j$, and $dx_j$ means the product of the differentials which correspond to functionally
    independent coefficients of $x_j$;

    \item[$\bullet$] variables $\psi$, $\Psi$, $\rho$, and $\tau$ with different indices are Grassmann variables or vectors or matrices
    correspondingly; if $\rho_j$ corresponds to the site $j=1,..,n$, then $\rho$ means
    vector $\{\rho_j\}_{j=1}^n$,  $d\rho=\prod d\rho_j$, and $d\rho_j$ means the product of the differentials which correspond the components
    (for vectors) or entries (for matrices)  taken into the lexicographic order; 
   
  \end{itemize}

\begin{fleqn}[18pt]
    \begin{align}\label{a_pm}
    \bullet \,\,\, & a_{\pm}=\dfrac{-iE\pm\sqrt{4-E^2}}{2},\quad c_{\pm}=1+a_{\pm}^{-2},\quad  c_0=
\sqrt{4-E^2}=2\pi\rho(E);\\
    \label{L_pm}
    &L=\hbox{diag}\,\{1,-1\},\quad L_{\pm}=\hbox{diag}\,\{a_+,a_-\};
    \end{align}
  
  \begin{align}\label{U}
   \bullet \,\,\, &
  \mathring{U}(2)=U(2)/U(1)\times U(1),\quad \mathring{U}(1,1)=U(1,1)/U(1)\times U(1),
  \end{align}
   \hskip 1cm where $U(p)$ is a group of $p\times p$ unitary matrices, and $U(1,1)$ is a group of $2\times 2$ hyperbolic 
   
   \noindent \hskip 1cm matrices $S$ such that $S^*LS=L$;
 \end{fleqn}

\begin{fleqn}[18pt]

\begin{align}\label{L_cal}
\bullet \,\,\,\mathcal{L}_\pm(E)&=\Big\{r\Big(-iE/2\pm\sqrt{4-E^2}/2\Big)|r\in [0,+\infty)\Big\};
\end{align}

\begin{align}
\bullet \,\,\,\tilde{\beta}=c_0^2\,\beta \label{beta_til};
\end{align}

\begin{align}\label{Z_12}
\bullet \,\,\,\ 
&Z_1=E_1\cdot I+\dfrac{1}{N}\Lambda_{\kappa,y_1}, \quad\, Z_2=E_2\cdot I+\dfrac{1}{N}\Lambda_{\kappa,y_2},\\ \notag
&\Lambda_{\kappa,y}=\left(\begin{matrix}-i\kappa&-y\\ y& i\kappa\end{matrix}\right),\quad \sigma=\left(\begin{matrix}
0&1\\
-1&0
\end{matrix}\right).
\end{align}
\end{fleqn}

\section{Integral representations}\label{s:2}
%
In this section we obtain an integral representation for $Z_{\beta nW}(\kappa,z_1,z_2)$ of (\ref{Z}).
\begin{proposition}\label{p:int_repr}The determinant ratio $Z_{\beta nW}(\kappa,z_1,z_2)$ of (\ref{Z}) can be written as follows:
\begin{align}\notag
&Z_{\beta nW}(\kappa,z_1,z_2)=C_{n,W}\displaystyle\int 
\exp\big\{-
i\sum\limits_{j=1}^n \Tr (LY_j+\delta_{j1}\sum_{a=1}^MLQ_a)Z_2-i\sum_{a=1}^M\gamma_a\Tr (LQ_a)\sigma\big\}
\end{align}
\begin{align} \notag
 &\times\exp\Big\{-
\dfrac{1}{2}\sum\limits_{j,k=1}^nJ_{jk}\Tr (LY_j+\delta_{j1}\sum_{a=1}^MLQ_a)(LY_k+\delta_{k1}\sum_{a=1}^MLQ_a)-\dfrac{1}{2}\sum\limits_{j,k=1}^n(J^{-1})_{jk}\Tr X_jX_k\Big\}\\
&\times \det\widetilde{ \mathcal D}\cdot \prod\limits_{j=1}^n\dfrac{\mdet^W (X_j-iZ_1)\mdet^W Y_j }{\mdet^2 Y_j}\cdot \prod\limits_{a=1}^M\dfrac{\mdet (X_1-iZ_1+i\gamma_a\sigma)}{\mdet (X_1-iZ_1)\cdot \mdet\, Y_1} \,\,dXdYdQ ,\label{sup}
\end{align}
where 
\[
\widetilde{\mathcal{D}}=\Big\{J^{-1}_{jk}\mathbf{1}_4-\delta_{jk}\Big((X_j-iZ_1)^{-1}\otimes (LY_j)+\delta_{j1}\sum\limits_{a=1}^M(X_j-iZ_1+i\gamma_a\sigma)^{-1}\otimes (LQ_a)\Big)\Big\}_{j,k=1}^n,
\]
\begin{align*}
Q_a=\left(\begin{matrix}
\bar\phi_{11a}\phi_{11a} & \bar\phi_{11a}\phi_{21a}\\ 
\bar\phi_{21a}\phi_{11a}&\bar\phi_{21a}\phi_{21a}
\end{matrix}\right),\quad dQ=\prod\limits_{a=1}^M\prod\limits_{l=1}^2\dfrac{d\Re \phi_{l1a}\, d\Im \phi_{l1a}}{\pi},
\end{align*}
for complex $\phi_{l1a}$. $\{X_j\}_{j=1}^n$ are Hermitian  $2\times 2$ matrices with standard $dX_j$ , $\{Y_j\}_{j=1}^n$ are $2\times 2$ 
positive Hermitian matrices with $dY_j$ of  Proposition \ref{p:supboz}, and $Z_{1,2}$ are defined in (\ref{Z}), and
\[
C_{n,W}=\dfrac{\mdet^2 J\cdot (-1)^{nW}}{(2\pi^3)^{n}\big((W-1)!(W-2)!\big)^{n-1}((W-M-1)!(W-M-2)!)},
\]
\end{proposition}

\begin{proof} To simplify computation, we are going to present the detailed derivation for the case $M=1$. 
General case can be obtained similarly with minor modifications.

To obtain SUSY integral representation, it is useful to rewrite $Z_{\beta nW}$ in the more convenient form. Notice that if we set
\begin{equation}\label{P}
P:=P(E,\kappa,y)=\left(\begin{matrix}
\dfrac{\kappa}{N}-i(H_N-E)&-i(\Gamma-\dfrac{y}{N})\\
-i(\Gamma-\dfrac{y}{N})&\dfrac{\kappa}{N}+i(H_N-E)
\end{matrix}\right),\quad T=\dfrac{1}{\sqrt{2}}\left(\begin{matrix}1&i\\ i&1\end{matrix}\right)\otimes I_N,
\end{equation}
then
\[
\mdet(TPT)=\mdet\Big\{(\mathcal{H}+i\Gamma-E-\dfrac{iy}{N})(\mathcal{H}-i\Gamma-E+\dfrac{iy}{N})+\dfrac{\kappa^2}{N^2}\Big\}.
\]
Hence
\begin{equation*}
Z_{\beta nW}(\kappa,z_1,z_2)=\mathbb{E}\left[\dfrac{\mdet\,P_1}{\mdet\,P_2}\right],
\end{equation*}
where
\begin{equation}\label{P_1,2}
P_1=P(E_1,\kappa,y_1),\quad P_1=P(E_2,\kappa,y_2).
\end{equation}
Such transformation is needed since we want $P_1$, $P_2$ to have positive real part.

Introduce complex and Grassmann fields:
\begin{align*}
\Phi_l=\{\phi_{lj}\}^t_{j=1,..,n},&\quad \phi_{l j}=(\phi_{l j 1}, \phi_{l j 2},\ldots,
\phi_{ l j W}),\quad l=1,2,
\quad -\quad \hbox{complex},\\
\Psi_l=\{\psi_{l j}\}^t_{j=1,..,n}, &\quad \psi_{ l j}=(\psi_{l j 1}, \psi_{l j 2},\ldots,
\psi_{l j W}),\quad l=1,2, \quad -\quad \hbox{Grassmann}.
\end{align*}
Since $P_1$, $P_2$ have positive real part, using (\ref{G_C}) -- (\ref{G_Gr}) (see Appendix) we can rewrite $\mdet\, P_1$ and $\mdet\, P_2$ of  (\ref{P}) -- (\ref{P_1,2}) and get
\begin{equation*}
\begin{array}{c}
Z_{\beta nW}(\kappa,z_1,z_2)
=\pi^{-2Wn}\mathbf{E}\Big\{ \displaystyle{\int}\exp\{-\Psi_1^+(\frac{\kappa}{N}+iE_1-iH_N)\Psi_1
-\Psi_2^+(\frac{\kappa}{N}-iE_1+iH_N)\Psi_2\}\\
\times\exp\{-\Phi_1^+(\frac{\kappa}{N}+iE_2-iH_N)\Phi_1-\Phi_2^+
(\frac{\kappa}{N}-iE_2+iH_N)\Phi_2\}\\
\times \exp\{i\Psi_1^+(\Gamma+\frac{y_1}{N})\Psi_2+i\Psi_2^+(\Gamma+\frac{y_1}{N})\Psi_1+i\Phi_1^+(\Gamma+\frac{y_2}{N})\Phi_2+i\Phi_2^+(\Gamma+\frac{y_2}{N})\Phi_1\}d\Phi d\Psi\Big\}
\end{array}
\end{equation*}
\begin{equation*}
\begin{array}{c}
=\displaystyle\int  \exp\Big\{-(\frac{\kappa}{N}+iE_1)\Psi_1^+\Psi_1-(\frac{\kappa}{N}-iE_1)\Psi_2^+\Psi_2-(\frac{\kappa}{N}+iE_2)\Phi_1^+\Phi_1-(\frac{\kappa}{N}-iE_1)\Phi_2^+\Phi_2\Big\}\\
\times\exp\{\dfrac{iy_1}{N}\Big(\Psi_1^+\Psi_2+\Psi_1^+\Psi_2\Big)+\dfrac{iy_2}{N}\Big(\Phi_1^+\Phi_2+\Phi_1^+\Phi_2\Big)\}\\
\times\exp\{i\gamma(\bar\psi_{111}\psi_{211}+\bar\psi_{211}\psi_{111}+\bar\phi_{111}\phi_{211}+\bar\phi_{211}\phi_{111})\Big\}\\
\times\mathbf{E}\Big\{\exp\Big\{\sum\limits_{j\le k}\sum\limits_{\alpha, \gamma}
\Big(i\Re H_{jk,\alpha\gamma}\chi^+_{jk,\alpha\gamma}
-\Im H_{jk,\alpha\gamma}\chi^-_{jk,\alpha\gamma}\Big)\Big\}\Big\}d\Phi d\Psi\,\,,
\end{array}
\end{equation*}
where 
\begin{align*}
&\chi^{\pm}_{jk,\alpha\gamma}=\eta_{jk,\alpha\gamma}\pm \eta_{kj,\gamma\alpha},\\
&\eta_{jk,\alpha\gamma}=\overline{\psi}_{1 j \alpha}\psi_{1 k \gamma}-
\overline{\psi}_{2j \alpha}\psi_{2 k \gamma}+\overline{\phi}_{1 j \alpha}\phi_{1 k \gamma}-
\overline{\phi}_{2j\alpha}\phi_{2 k \gamma},\\
&\eta_{jj,\alpha\alpha}=(\overline{\psi}_{1 j \alpha}\psi_{1 j \alpha}-
\overline{\psi}_{2j\alpha}\psi_{2 j \alpha}+\overline{\phi}_{1 j \alpha}\phi_{1 j \alpha}-
\overline{\phi}_{2j \alpha}\phi_{2 j \alpha})/2.
\end{align*}
Averaging over (\ref{pr_l}), we get
\begin{align*}
&Z_{\beta nW}(\kappa,z_1,z_2)=\pi^{-2Wn}\int d\Phi d\Psi\,  \exp\{\dfrac{iy_1}{N}\Big(\Psi_1^+\Psi_2+\Psi_1^+\Psi_2\Big)+\dfrac{iy_2}{N}\Big(\Phi_1^+\Phi_2+\Phi_1^+\Phi_2\Big)\}\\
&\times \exp\Big\{-(\frac{\kappa}{N}+iE_1)\Psi_1^+\Psi_1-(\frac{\kappa}{N}-iE_1)\Psi_2^+\Psi_2-(\frac{\kappa}{N}+iE_2)\Phi_1^+\Phi_1-(\frac{\kappa}{N}-iE_1)\Phi_2^+\Phi_2\Big\}\\
&\times\exp\{i\gamma(\bar\psi_{111}\psi_{211}+\bar\psi_{211}\psi_{111}+\bar\phi_{111}\phi_{211}+\bar\phi_{211}\phi_{111})\Big\}\\
&\times\exp\Big\{-\sum\limits_{j<k}\sum\limits_{\alpha,\gamma} J_{jk}\,\,
\eta_{jk,\alpha\gamma}\eta_{kj,\gamma\alpha}-\frac{1}{2}\sum\limits_{j, \alpha} J_{jj}\,\,
\eta_{jj,\alpha\alpha}^2\Big\}.
\end{align*}
Define 
\begin{align*}
Q=\left(\begin{matrix}
\bar\phi_{111}\phi_{111} & \bar\phi_{111}\phi_{211}\\ 
\bar\phi_{211}\phi_{111}&\bar\phi_{211}\phi_{211}
\end{matrix}\right)
\end{align*}
and set
\begin{align*}
& \tilde{Y}_j=\left(
\begin{array}{ll}
\phi_{1j}^+\phi_{1j}& \phi_{1j}^+\phi_{2j}\\
\phi_{2j}^+\phi_{1j}& \phi_{2j}^+\phi_{2j}
\end{array}
\right), j\ne 1,\quad \tilde{Y}_1=\left(
\begin{array}{ll}
\sum\limits_{\alpha=2}^{W}\bar\phi_{11\alpha}\phi_{11\alpha}& \sum\limits_{\alpha=2}^{W}\bar\phi_{11\alpha}\phi_{21\alpha}\\
\sum\limits_{\alpha=2}^{W}\bar\phi_{21\alpha}\phi_{11\alpha}& \sum\limits_{\alpha=2}^{W}\bar\phi_{21\alpha}\phi_{21\alpha}
\end{array}\right)\\
&\tilde{X}_j=\left(
\begin{array}{ll}
\psi_{1j}^+\psi_{1j}& \psi_{1j}^+\psi_{2j}\\
\psi_{2j}^+\psi_{1j}& \psi_{2j}^+\psi_{2j}
\end{array}
\right).
\end{align*}
Thus, 
\begin{eqnarray}\label{G_av}
&Z_{\beta nW}(\kappa,z_1,z_2)=\pi^{-2Wn}\intd d\Phi d\Psi\,\, \exp\Big\{-i\sum\limits_{j=1}^n \Tr \tilde{X}_jLZ_1-
i\sum\limits_{j=1}^n \Tr (L\tilde{Y}_j+\delta_{j1}LQ)Z_2\Big\}\\ \notag
&\times\exp\Big\{i\gamma(\bar\psi_{111}\psi_{211}+\bar\psi_{211}\psi_{111})-i\gamma\Tr (LQ)\sigma\Big\}
\end{eqnarray}
\begin{eqnarray}\notag
&\times \exp\Big\{\dfrac{1}{2}\sum\limits_{j,k=1}^nJ_{jk}\Tr (\tilde{X}_jL)(\tilde{X}_kL)-
\dfrac{1}{2}\sum\limits_{j,k=1}^nJ_{jk}\Tr (L\tilde{Y}_j+\delta_{j1}LQ)(L\tilde{Y}_k+\delta_{k1}LQ)\Big\}\\ \notag
&\times \exp\Big\{-\sum\limits_{j,k=1}^nJ_{jk}
\big(\overline{\psi}_{1j}\psi_{1k}(\overline{\phi}_{1k}\phi_{1j}-
\overline{\phi}_{2k}\phi_{2j})+\overline{\psi}_{2j}\psi_{2k}(\overline{\phi}_{2k}\phi_{2j}-
\overline{\phi}_{1k}\phi_{1j})\big)\Big\},
\end{eqnarray}
where $L$, $Z_{1,2}$, $\sigma$ are defined in (\ref{L_pm}) and (\ref{Z_12}).

Using the standard Hubbard-Stratonovich transformation, we obtain
\begin{multline}\label{Hub}
\big(2\pi^2\big)^{n}\mdet^{2}J\cdot\exp\Big\{\dfrac{1}{2}\sum\limits_{j,k=1}^nJ_{jk}\Tr (\tilde{X}_jL)(\tilde{X}_kL)\Big\}\\ 
=\int \exp\Big\{-\dfrac{1}{2}\sum\limits_{j,k=1}^n(J^{-1})_{jk}\Tr X_jX_k+
\sum\limits_{j=1}^n\Tr X_j\big(\tilde{X}_jL\big)\Big\}dX,
\end{multline}
where $X_j$ are $2\times 2$ Hermitian matrices with the standard measure $dX_j$. 

Substituting (\ref{Hub}) to (\ref{G_av}) and integrating over $d\Psi$ (see (\ref{G_Gr})), we get
\begin{eqnarray}\notag
&Z(\kappa, z_1, z_2)=\dfrac{\mdet^{-2}J}{\big(2\pi^{2(1+W)}\big)^{n}}
\int \exp\Big\{-\dfrac{1}{2}\sum\limits_{j,k=1}^n(J^{-1})_{jk}\Tr X_jX_k-i\gamma\Tr (LQ)\sigma\Big\}\\ \label{G_M}
&\times\exp\Big\{-
i\sum\limits_{j=1}^n \Tr (L\tilde{Y}_j+\delta_{j1}LQ)Z_2-
\dfrac{1}{2}\sum\limits_{j,k=1}^nJ_{jk}\Tr (L\tilde{Y}_j+\delta_{j1}LQ)(L\tilde{Y}_k+\delta_{k1}LQ)\Big\}\\ \notag
&\times  \mdet\, M\cdot\,d\Phi\,
 \, dX
\end{eqnarray}
with $M=M^{(1)}-M^{(2)}$. Here $M^{(1)}$ and $M^{(2)}$ are  $2Wn\times 2Wn$ matrices with entries
\begin{align}
\notag
M^{(1)}_{l j\alpha, l' k\gamma}&=\delta_{jk}\delta_{\alpha\gamma} (C_{j\alpha})_{ll'}L_{ll},\quad
j,k=1,..,n,\,\,\alpha,\gamma=1,\ldots, W,\,\,l,l'=1,2, \\ \label{M_1,2}
M^{(2)}_{l j\alpha, l' k\gamma}&=J_{jk}\delta_{ll'} L_{ll}\sum\limits_{\nu=1}^2\varphi_{\nu j\alpha}\overline{\varphi}_{\nu k\gamma} L_{\nu\nu}
\end{align}
with
\begin{equation}\label{C}
C_{j\alpha}=\begin{cases}
X_1-iZ_1+i\gamma\sigma,&j=\alpha=1;\\
X_j-iZ_1,&\text{otherwise}.
\end{cases}
\end{equation}
We can rewrite
\begin{equation*}
\mdet M=\mdet M^{(1)}\cdot \mdet \Big(1-\big(M^{(1)}\big)^{-1}M^{(2)}\Big)=:\mdet M^{(1)}\cdot \mdet \Big(1-\mathcal{M}\Big)
\end{equation*}
with
\begin{equation}\label{Mcal}
\mathcal{M}_{l j\alpha, l' k\gamma }=J_{jk}(C_{j\alpha})_{ll'}^{-1}\sum\limits_{\nu=1}^2\varphi_{\nu j\alpha}\overline{\varphi}_{\nu k\gamma} L_{\nu\nu}.
\end{equation}
Note that $\mathcal{M}=AB$, where
\begin{align}\notag
A_{l j\alpha,l' k\sigma}&=J_{jk}(C_{j\alpha})^{-1}_{ll'}\,\varphi_{\sigma j\alpha}, \quad j,k\in \Lambda,\,\,\alpha,\gamma=1,\ldots, W,\,\,l,l',\sigma=1,2,\\ \label{AB}
B_{l j\sigma, l' k\alpha }&=\delta_{jk}\delta_{ll'}L_{\sigma\sigma}\,\overline{\varphi}_{\sigma k\alpha}. 
\end{align}
Therefore, using that $\mdet (1-AB)=\mdet (1-BA)$, (\ref{Mcal}), and (\ref{AB}), we get
\begin{equation}\label{ch_M}
\mdet \Big(1-\mathcal{M}\Big)=\mdet \Big(1-BA\Big)=:\mdet \Big(1-\tilde{\mathcal{M}}\Big),
\end{equation} 
where
\begin{align*}
\tilde{\mathcal{M}}_{l j\sigma, l'k\sigma' }&=\sum\limits_{p,\alpha,\nu}B_{l j\sigma, \nu p\alpha} A_{\nu p\alpha, l' k\sigma' }=
J_{jk}\sum\limits_{\alpha=1}^W(C_{j\alpha})^{-1}_{ll'}\overline{\varphi}_{\sigma j\alpha}\varphi_{\sigma' j\alpha}L_{\sigma\sigma}
\end{align*}
\begin{align}
\notag
&=\begin{cases}
J_{jk}(X_j-iZ_1)^{-1}_{ll'} (L\tilde{Y}_j)_{\sigma\sigma'},&j>1\\
J_{1k}(X_1-iZ_1)^{-1}_{ll'} (L\tilde{Y}_1)_{\sigma\sigma'}+J_{1k}\left(X_1-iZ_1+i\gamma\sigma\right)^{-1}_{ll'} (LQ)_{\sigma\sigma'},& j=1
\end{cases}
\end{align}
Here we substituted (\ref{C}).

This yields
\begin{align*} \notag
&\mdet (1-\tilde{\mathcal{M}})=\mdet \big\{\delta_{jk}-J_{jk}\Big((X_j-iZ_1)^{-1}\otimes (L\tilde{Y}_j)+\delta_{j1}\left(X_1-iZ_1+i\gamma\sigma\right)^{-1}\otimes (LQ)\Big)\big\}\\ 
&=\mdet^4 J\cdot \mdet\big\{J^{-1}_{jk}\mathbf{1}_4-\delta_{jk}\Big((X_j-iZ_1)^{-1}\otimes (\tilde{Y}_jL)+\delta_{j1}\left(X_1-iZ_1+i\gamma\sigma\right)^{-1}\otimes (LQ)\Big)\big\}.
\end{align*}
Besides,
\begin{equation}\label{mdet1}
\mdet\, M^{(1)}=(-1)^{nW}\dfrac{\mdet (X_1-iZ_1+i\gamma\sigma)}{\mdet (X_1-iZ_1)}\cdot \prod\limits_{j=1}^n\mdet^W (X_j-iZ_1).
\end{equation}
Now substituting (\ref{M_1,2}) -- (\ref{Mcal}) and (\ref{ch_M}) -- (\ref{mdet1}) to (\ref{G_M}) and
applying the bosonization formula (see Proposition \ref{p:supboz}), we obtain (\ref{sup}) which finishes the proof for the case $M=1$.
\end{proof}


\section{Derivation of  the sigma-model approximation}\label{s:3}
\subsection{Proof of Theorem \ref{thm:sigma_mod}}

Again we are going to concentrate on the case $M=1$. 

Let $\beta$ and $n$ be fixed, and  $W\to \infty$. 
Defining $n\times n$ matrix $R$ as
\begin{equation*}
J^{-1}=W\big(1-\dfrac{\beta}{W}\triangle+\dfrac{\beta^2}{W^2}\triangle^2-\ldots\big)=:W\big(1-\dfrac{\beta}{W}\triangle+\dfrac{1}{W^2}R\big),
\end{equation*}
putting $B_j=W^{-1}LY_j$, and  shifting $X_j-iZ_1\to X_j$, we can rewrite (\ref{sup}) of Proposition \ref{p:int_repr} as
\begin{eqnarray} \notag
&Z_{\beta nW}(\kappa,z_1,z_2)=C^{(1)}_{W, n}\displaystyle\int  dX\,dB\, dQ\cdot\det \,D\cdot \prod\limits_{j=1}^n\dfrac{\mdet^W X_j\cdot \mdet^W B_j }{\mdet^2 B_j}\cdot \dfrac{\mdet (X_1+i\gamma\sigma)}{\mdet\, X_1\cdot \mdet\, B_1}\\ \label{sup1}
&\times \exp\Big\{-
 \Tr (LQ)(iZ_2+i\gamma\sigma)-\dfrac{W}{2}\sum\limits_{j=1}^n \Big(\Tr (B_j+iZ_2)^2+\Tr (X_j+iZ_1)^2\Big)\Big\}
\end{eqnarray}
\begin{eqnarray}\notag 
&\times \exp\Big\{-\dfrac{1}{2W}\Tr (LQ)^2-\Tr B_1 (LQ)+\dfrac{\beta}{W}\Tr (B_1-B_2)(LQ)+\dfrac{\beta}{2W^2}\Tr (LQ)^2\Big\} 
 \\ \notag
&\times\exp\Big\{
\dfrac{\beta}{2}\sum\limits_{j=1}^{n-1}\Big(\Tr (B_j-B_{j+1})^2-\Tr (X_j-X_{j+1})^2\Big)+\dfrac{1}{2W}\sum\limits_{j, k}
R_{jk}\,\Tr (X_j+iZ_1)(X_k+iZ_1)\Big\},
\end{eqnarray}
where 
\[
\mathcal{D}=\Big\{(1-\dfrac{\beta}{W}\triangle+\dfrac{1}{W^2}R)_{jk}\mathbf{1}_4-\delta_{jk}\Big(X_j^{-1}\otimes B_j+\dfrac{\delta_{j1}}{W}(X_1+i\gamma\sigma)^{-1}\otimes (LQ)\Big)\Big\}_{j,k=1}^{n}
\]
and 
\begin{align*}
C^{(1)}_{W, n}&=\dfrac{\mdet^2 J\cdot W^{8n}\cdot W^{2(W-2)(n-1)}\cdot W^{2(W-3)} \cdot e^{-Wn\Tr Z_2^2/2}}{(2\pi^3)^{n}\big((W-1)!(W-2)!\big)^{n-1}((W-2)!(W-3)!)}\\ \notag
&=\dfrac{ W^{4n}\cdot e^{2nW-Wn\Tr Z_2^2/2}}{(2\pi^2)^{2n}}\cdot \Big(1+O\big(W^{-1}\big)\Big).
\end{align*}
Change the variables to
\begin{align*}
X_j&=U_j^*\hat{X}_jU_j,\,\,\,\,\,\, \hat{X}_j=\hbox{diag}\,\{x_{j,1},x_{j,2}\},\,\,\,\,  U_j\in \mathring{U}(2),
\quad\,\,\, x_{j,1}, x_{j,2}\in \mathbb{R},\\ \notag
B_j&=S_j^{-1}\hat{B}_jS_j,\,\,\,\, \hat{B}_j=\hbox{diag}\,\{b_{j,1},b_{j,2}\},\quad
S_j\in \mathring{U}(1,1),
\,\,\, b_{j,1}\in \mathbb{R}^+,\, b_{j,2}\in \mathbb{R}^-.
\end{align*}
The Jacobian of such a change is
\[
2^{n}(\pi/2)^{2n}\prod\limits_{j=1}^n(x_{j,1}-x_{j,2})^2
\prod\limits_{j=1}^n(b_{j,1}-b_{j,2})^2.
\]
This and (\ref{sup1}) yield
\begin{align}\label{G_last}
&Z_{\beta nW}(\kappa,z_1,z_2)=C^{(2)}_{W, n}
\int  dS\,dU\,dQ
\int d x\int_{\mathbb{R}_+^{n}\times \mathbb{R}_-^{n}} d b
\cdot \prod\limits_{j=1}^n\dfrac{(x_{j,1}-x_{j,2})^2
(b_{j,1}-b_{j,2})^2}{b_{j,1}^2b_{j,2}^2}\\ \notag
&\times \dfrac{\mdet (\hat X_1+i\gamma U_1\sigma U_1^*)}{x_{1,1}x_{1,2}\cdot b_{1,1}b_{1,2}}
\cdot \mdet \,\mathcal{D}\cdot
\exp\Big\{-W\sum\limits_{j=1}^n\sum\limits_{l=1}^2\left(f(x_{j,l})+f(b_{j,l})\right)\Big\}\\ \label{exp1}
&\times \exp\Big\{-
 \Tr (LQ)(iZ_2+i\gamma\sigma)-\Tr S_1^{-1}\hat B_1S_1 (LQ)\Big\}\\ \notag
&\times \exp\Big\{-\dfrac{1}{2W}\Tr (LQ)^2+\dfrac{\beta}{W}\Tr (S_1^{-1}\hat B_1S_1-S_2^{-1}\hat B_2S_2)(LQ)+\dfrac{\beta}{2W^2}\Tr (LQ)^2\Big\}
\end{align}
\begin{align}\label{exp2}
&\times \exp\Big\{
\dfrac{\beta}{2}\sum\limits_{j=2}^n\Big(\Tr (S_j^{-1}\hat{B}_jS_j-S_{j-1}^{-1}\hat{B}_{j-1}S_{j-1})^2-\Tr (U_j^*\hat{X}_jU_j-U_{j-1}^*\hat{X}_{j-1}U_{j-1})^2\Big)\Big\}\\ \notag
&\times\exp\Big\{\dfrac{1}{2W}\sum\limits_{j, k}
R_{jk}\,\Tr (U_j^*\hat{X}_jU_j+iZ_1)(U_k^*\hat{X}_kU_k+iZ_1)\Big\}\\ \label{exp3}
&\times\exp\Big\{-\dfrac{1}{n}\sum\limits_{j=1}^n\Big(\Tr U_j^*\hat{X}_jU_j\Big(\begin{matrix}\kappa+ix_1&-iy_1\\ iy_1&-\kappa+ix_1\end{matrix}\Big)+\Tr S_j^{-1}\hat{B}_jS_j
\Big(\begin{matrix}\kappa+ix_2&-iy_2\\ iy_2&-\kappa+ix_2\end{matrix}\Big)\Big)\Big\},
\end{align}
where 
\begin{multline}\label{D}
\det \mathcal{D}
=\det
\Big\{\delta_{jk}\big(\mathbf{1}_4-\hat{X}_j^{-1}\otimes \hat{B}_j\big)-\dfrac{\delta_{jk}\delta_{j1}}{W}(\hat X_1+i\gamma U_1\sigma U_1^*)^{-1}\otimes \big(S_1(LQ)S_1^{-1}\big)\\+\dfrac{1}{W}\Big(-\beta\triangle+\dfrac{1}{W}R \Big)_{jk}\cdot U_jU_k^*\otimes S_jS_k^{-1}\Big\}_{j,k=1}^n,
\end{multline}
\begin{align*}
C^{(2)}_{W, n}&=2^{n}(\pi/2)^{2n}\cdot e^{Wn(\Tr Z_1^2+\Tr Z_2^2)/2-Wn (2+E^2)} \cdot C^{(1)}_{W, n}\\ \notag
&=\dfrac{ W^{4n}\cdot e^{2Ex_1}}{2^{3n}\pi^{2n}}\cdot \Big(1+O\big(W^{-1}\big)\Big),\\ 
\notag
f(x)&=x^2/2+iE x-\log x -(2+E^2)/4.
\end{align*}
The constant in $f(x)$ is chosen in such a way  that $\Re f(a_\pm)=0$. Measures $dU_j$, $dS_j$ in (\ref{G_last}) are the Haar measures over $\mathring{U}(2)$ and $\mathring{U}(1,1)$ correspondingly.

Also it is easy to see that for $|E|\le\sqrt{2}$ we can deform the contours of integration as 
\begin{itemize}
\item for $x_{j,1}$, $x_{j,2}$ to $-iE/2+\mathbb{R}$;
\item for $b_{j,1}$ to $\mathcal{L}_+(E)$ of (\ref{L_cal});
\item for $b_{j,2}$ to $\mathcal{L}_-(E)$ of (\ref{L_cal}).
\end{itemize}

To prove Theorem \ref{thm:sigma_mod}, we are going to integrate (\ref{G_last}) over the ``fast'' variables: $\{x_{j,l}\}, \{b_{j,l}\}$, $l=1,2$, $j=1,..,n$.
The first step is the following lemma:
\begin{lemma}\label{l:s_point}
The integral (\ref{G_last}) over $\{x_{j,l}\}, \{b_{j,l}\}$, $l=1,2$, $j=1,..,n$ can be restricted to the integral over the $W^{-(1-\kappa)/2}$-neighbourhoods (with a  small $\kappa>0$) of the points
\begin{itemize}
\item[I.] $x_{j,1}=a_+$, $x_{j,2}=a_-$ or $x_{j,1}=a_-$, $x_{j,2}=a_+$, $b_{j,1}=a_+$, $b_{j,2}=a_-$ for any $j=1,..,n$;
\item[II.] $x_{j,1}=x_{j,2}=a_+$, $b_{j,1}=a_+$, $b_{j,2}=a_-$ for any $j=1,..,n$;
\item[III.] $x_{j,1}=x_{j,2}=a_-$, $b_{j,1}=a_+$, $b_{j,2}=a_-$ for any $j=1,..,n$.
\end{itemize}
Moreover, the contributions of the points II and III are $o(1)$, as $W\to\infty$.
\end{lemma}
\begin{proof}
The proof of the first part of the lemma is straightforward and  based on the fact that $\Re f(z)$ for $z=x-iE/2$, $x\in \mathbb{R}$ has 
two global minimums at $z=a_\pm$, and for $z\in\mathcal{L}_\pm(E)$ has one global minimum at $z=a_\pm$.

To prove the second part of the lemma, consider the neighbourhood of the point II (the point III can be treated in a similar way). 
Change the variables as
\begin{align*}
\begin{array}{ll}
x_{j,1}=a_++{\tilde{x}_{j,1}}/{\sqrt{W}}, &x_{j,2}=a_++\tilde{x}_{j,2}/{\sqrt{W}},\\ 
b_{j,1}=a_+\big(1+{\tilde{b}_{j,1}}/{\sqrt{W}}\big), &b_{j,2}=a_-\big(1+{\tilde{b}_{j,2}}/{\sqrt{W}}\big).
\end{array}
\end{align*}
This gives the Jacobian $(-1)^{n} W^{-2n}$ and also the additional $W^{-n}$ since
\begin{align*}
x_{j,1} -x_{j,2}=(\tilde{x}_{j,1}-\tilde{x}_{j,2})/\sqrt{W}.
\end{align*}
Together with $C^{(2)}_{W, n}$ this gives $W^{n}$ in front of the integral (\ref{G_last}).
In addition, expanding $f$ into the series, we get
\begin{align}\label{f_exp}
&f(x_{j,l})=f(a_+)+\dfrac{c_+}{2}\frac{\tilde{x}_{j,l}^2}{W}-\dfrac{1}{2a_+^3}\frac{\tilde{x}_{j,l}^3}{W^{3/2}}+O\Big(\dfrac{\tilde{x}_{j,l}^4}{W^{2}}\Big), \quad l=1,2\\ \notag
&f(b_{j,1})=f(a_+)+\dfrac{a_+^2c_+}{2}\cdot \dfrac{\tilde{b}_{j,1}^2}{W}-\dfrac{1}{2}\cdot \dfrac{\tilde{b}_{j,1}^3}{W^{3/2}}+O\Big(\dfrac{\tilde{b}_{j,1}^4}{W^{2}}\Big),\\ \notag
& f(b_{j,2})=f(a_-)+\dfrac{a_-^2c_-}{2}
\cdot \dfrac{\tilde{b}_{j,2}^2}{W}-\dfrac{1}{2}\cdot \dfrac{\tilde{b}_{j,2}^3}{W^{3/2}}+O\Big(\dfrac{\tilde{b}_{j,2}^4}{W^{2}}\Big),
\end{align}
where
\begin{equation}\label{c_pm}
c_\pm=1+a_\pm^{-2}, \quad f(a_+)=-f(a_-)\in i\mathbb{R}.
\end{equation}

We are going to compute the leading order of the integral over $\{\tilde{x}_{j,l}\}, \{\tilde{b}_{j,l}\}$, $l=1,2$, $j=1,..,n$. To this end,  we leave 
the quadratic part of $f$ (see (\ref{f_exp})) in the exponent,  expand everything else into the series of $\tilde{x}_{j,l}/\sqrt{W}, \tilde{b}_{j,l}/\sqrt{W}$ 
around the saddle-point $\tilde{x}_{j,l}=\tilde{b}_{j,l}=0$, and compute the Gaussian integral of each term of this expansion.
We are going to prove that all this terms are $o(1)$.

Indeed, consider the expansion of the diagonal elements of $\mathcal{D}$ of (\ref{D}):
\begin{align}\notag
&d_{j,l1}=1-x_{j,l}^{-1}b_{j,1}=(\tilde{x}_{j,l}/a_+-\tilde{b}_{j,1})/\sqrt{W}+O\big(W^{-1+2\kappa}\big),\\
&d_{j,l2}=1-x_{j,l}^{-1}b_{j,2}=c_--({\tilde{x}_{j,l}/a_+-\tilde{b}_{j,2}})/a_-^2\sqrt{W}+O\big(W^{-1+2\kappa}\big),\quad l=1,2.\label{d_exp1_1}
\end{align}
If we rewrite the determinant of $\mathcal{D}$  in a standard way, then each summand has strictly one element from each row and column.
Because of (\ref{d_exp1_1}), each element in the rows $(j,11)$ and $(j,21)$ has at least $W^{-1/2}$, and so the expansion of $\mdet \,\mathcal{D}$ starts from $W^{-n}$. Moreover, to obtain $W^{-n}$ (i.e. non-zero contribution) we must consider the summands of the determinant expansion that have only diagonal elements $d_{j,ls}$ (since non-diagonal elements of $\mathcal{D}$ are $O(W^{-1})$ or less), and furthermore only the first terms in the expansions (\ref{d_exp1_1}) and all other function
in (\ref{G_last}). Thus we get
\begin{align}\label{++}
C\cdot \Big\langle\prod\limits_{j=1}^n \dfrac{\tilde{x}_{j,1}/a_+-\tilde{b}_{j,1}}{\sqrt{W}}\cdot \dfrac{\tilde{x}_{j,2}/a_+-\tilde{b}_{j,1}}{\sqrt{W}} \cdot (\tilde{x}_{j,1}
-\tilde{x}_{j,2})^2\Big\rangle_{++}+o(1),
\end{align}
where
\begin{align*}
\Big\langle \cdot \Big\rangle_{++}=\int \Big(\cdot\Big) \exp\Big\{-\frac{1}{2}\sum_{j=1,..,n}\Big({c_+(\tilde x_{j,1}^2+\tilde x_{j,2}^2)}+
{a_+^2c_+ \tilde b_{j,1}^2}
+{a_-^2c_- \tilde b_{j,2}^2}\Big)\Big\}  d\tilde{x}\, d\tilde{b}.
\end{align*}
But it is easy to see that  the Gaussian integral in (\ref{++}) is zero, which completes the proof of the lemma.
\end{proof}

According to Lemma \ref{l:s_point} the main contribution to (\ref{G_last}) is given by the neighbourhoods of the saddle points $x_{j,1}=a_+$, $x_{j,2}=a_-$ or $x_{j,1}=a_-$, $x_{j,2}=a_+$. All such points can be obtained from each other by rotations of $U_j$, 
so we can consider only $x_{j,1}=a_+$, $x_{j,2}=a_-$ for all $j=1,..,n$. Similarly to the proof of Lemma \ref{l:s_point}, 
change variables as
\begin{align}\label{change}
\begin{array}{ll}
x_{j,1}=a_++{\tilde{x}_{j,1}}/{\sqrt{W}}, &x_{j,2}=a_-+{\tilde{x}_{j,2}}/{\sqrt{W}},\\ 
b_{j,1}=a_+\big(1+{\tilde{b}_{j,1}}/{\sqrt{W}}\big), &b_{j,2}=a_-\big(1+{\tilde{b}_{j,2}}/{\sqrt{W}}\big).
\end{array}
\end{align}
That slightly change the expansions (\ref{f_exp})  and (\ref{d_exp1_1}). We get
\begin{align}\label{f_exp2}
&f(x_{j,2})=f(a_-)+\dfrac{c_-}{2}\cdot \dfrac{\tilde{x}_{j,2}^2}{W}-\dfrac{1}{2a_-^3}\cdot \dfrac{\tilde{x}_{j,2}^3}{W^{3/2}}+O\Big(\dfrac{\tilde{x}_{j,2}^4}{W^{2}}\Big),
\end{align}
and 
\begin{align}\label{d_exp1}
&d_{j,11}=1-x_{j,1}^{-1}b_{j,1}=\dfrac{\tilde{x}_{j,1}/a_+-\tilde{b}_{j,1}}{\sqrt{W}}+
\dfrac{a_+\tilde{x}_{j,1}\tilde{b}_{j,1}-\tilde{x}_{j,1}^2}{a_+^2W}+\dfrac{\delta_{j1}}{W}T_{11,11}+O\big(W^{-3(1-\kappa)/2}\big),\\ \notag
&d_{j,22}=1-x_{j,2}^{-1}b_{j,2}=\dfrac{\tilde{x}_{j,2}/a_--\tilde{b}_{j,2}}{\sqrt{W}}+\dfrac{a_-\tilde{x}_{j,2}\tilde{b}_{j,2}-\tilde{x}_{j,2}^2}
{a_-^2W}+\dfrac{\delta_{j1}}{W}T_{22,22}+O\big(W^{-3(1-\kappa)/2}\big),
\end{align}
\begin{align}\notag
&d_{j,12}=1-x_{j,1}^{-1}b_{j,2}=c_+-\dfrac{\tilde{x}_{j,1}/a_+-\tilde{b}_{j,2}}{a_+^{2}\sqrt{W}}-\dfrac{a_+\tilde{x}_{j,1}\tilde{b}_{j,2}-\tilde{x}_{j,1}^2}{a_+^4W}
+\dfrac{\delta_{j1}}{W}T_{11,22}+O\big(W^{-3(1-\kappa)/2}\big), \\ \notag
&d_{j,21}=1-x_{j,2}^{-1}b_{j,1}=c_- -\dfrac{{\tilde{x}_{j,2}/a_--\tilde{b}_{j,1}}}{a_-^{2}\sqrt{W}}-\dfrac{a_-\tilde{x}_{j,2}\tilde{b}_{j,1}-\tilde{x}_{j,2}^2}{a_-^4W}
+\dfrac{\delta_{j1}}{W}T_{22,11}+O\big(W^{-3(1-\kappa)/2}\big),
\end{align}
where
\begin{align}\label{T}
T&=(\hat X_1+i\gamma U_1\sigma U_1^*)^{-1}\otimes \big(S_1(LQ)S_1^{-1}\big)\\ \notag
&=\Big(A^{-1}-\dfrac{1}{\sqrt{W}}A^{-1}
\left(\begin{matrix}
\tilde{x}_{1,1}&0\\
0&\tilde{x}_{1,2}
\end{matrix}\right)
A^{-1}\Big)\otimes \big(S_1(LQ)S_1^{-1}\big)+O(W^{-1+2\kappa})
\end{align}
with
\begin{equation}\label{A}
A=\big(\hat X_1+i\gamma U_1\sigma U_1^*\big)\Big|_{\tilde x_{1,1}=\tilde x_{1,2}=0}=-\frac{iE}{2}+\frac{c_0}{2}L+i\gamma U_1\sigma U_1^*.
\end{equation}
The change (\ref{change}) gives the Jacobian $W^{-2n}$, which together with $C^{(2)}_{W, n}$ gives $W^{2n}$ in front of the integral (\ref{G_last}).
Similarly to the proof of Lemma \ref{l:s_point} we are going to compute the leading order of the integral (\ref{G_last}) over $\{\tilde{x}_{j,l}\}, \{\tilde{b}_{j,l}\}$, $l=1,2$, $j=1,..,n$, and so  we leave 
the quadratic part of $f$ (see (\ref{f_exp}) and (\ref{f_exp2})) in the exponent,  expand everything else into the series of $\tilde{x}_{j,l}/\sqrt{W}, \tilde{b}_{j,l}/\sqrt{W}$ 
around the saddle-point $\tilde{x}_{j,l}=\tilde{b}_{j,l}=0$, and compute the Gaussian integral of each term of this expansion. We are going to prove, that the non-zero
contribution is given by the terms having at least $W^{-2n}$.

\begin{lemma}\label{l:det_exp} Formula (\ref{G_last}) can be rewritten as
\begin{align}\label{G_main}
&Z_{\beta nW}(\kappa,z_1,z_2)=(c_0/2\pi)^{2n}\cdot e^{E(x_1-x_2)}\int  dz\, d\tilde\rho\, d\tilde\tau \,d U\, d S\,d Q\\ &\times\exp\Big\{-\dfrac{1}{2}(Mz, z)+W^{1/2}(z,h^0)+
W^{-1/2}(z,h+\zeta/|\Lambda)|)\Big\} \notag \\ \notag 
&\times \exp\Big\{-\Tr (LQ)(iE/2+i\gamma\sigma)-\dfrac{c_0}{2}\Tr S_1^{-1}LS_1 (LQ)\Big\}
\\ \notag 
&\times \exp\Big\{-\Tr A^{-1}\tilde\rho_1S_1(LQ)S_1^{-1}\tilde\tau_1\Big\}\cdot \mdet \,A
\\ \notag
&\times \exp\Big\{{\beta}\sum\Tr\Big(U_j^*\tilde\rho_jS_j-U_{j-1}^*\tilde\rho_{j-1}S_{j-1}\Big) \Big(S_j^{-1}\tilde\tau_jU_j-
S_{j-1}^{-1}\tilde\tau_{j-1}U_{j-1}\Big)\Big\}\\ \notag
&\times \exp\Big\{\sum\Big(c_+ n_{j,12}+c_- n_{j,21}- n_{j,1}/{c_0a_+}+n_{j,2}/{c_0a_-}\Big)-\beta c_0^2\sum (v_j^2+t_j^2)\Big\}\\ \notag
&\times\exp\Big\{-\dfrac{c_0}{2n}\sum\limits_{j=1}^n\Big(\Tr U_j^*L U_j\Big(\begin{matrix}\kappa&-iy_1\\ iy_1&-\kappa\end{matrix}\Big)+
\Tr S_j^{-1} L S_j\Big(\begin{matrix}\kappa&-iy_2\\ iy_2&-\kappa\end{matrix}\Big)\Big)\Big\}
+o(1),
\end{align}
where $A$ is defined in (\ref{A}),
\begin{align}
\label{rt_tilde}
&\tilde\rho_j=\left(\begin{array}{cc}
\rho_{j,11}&\rho_{j,12}/\sqrt{W}\\
\rho_{j,21}/\sqrt{W}&\rho_{j,22}
\end{array}\right), \quad \tilde\tau_j=\left(\begin{array}{cc}
\tau_{j,11}&\tau_{j,12}/\sqrt{W}\\
\tau_{j,21}/\sqrt{W}&\tau_{j,22}
\end{array}\right)
\\
\notag
&n_{j,12}=\rho_{j,12}\tau_{j,12},\quad n_{j,21}=\rho_{j,21}\tau_{j,21},\\ \notag
&n_{j,1}=\rho_{j,11}\tau_{j,11},\quad\,\, n_{j,2}=\rho_{j,22}\tau_{j,22},\\ \notag
&z=(z_{j,11},z_{j,22},z_{j,12},z_{j,21})=(\tilde x_{j,1},\tilde x_{j,2},\tilde b_{j,1},\tilde b_{j,1}),
\end{align}
and 
\begin{align}\label{M}
&M=M_0+W^{-1}\tilde{M}\\
\label{M_0}
&(M_0 z,z)=\sum\limits_{j=1,..,n}\Big(c_+\tilde x_{j,1}^2+c_-\tilde x_{j,2}^2+a_+^2c_+\tilde b_{j,1}^2+a_-^2c_-\tilde b_{j,2}^2\Big) 
\end{align}
\begin{align}
\label{tilde_M}
&(\tilde{M}z,z)=-2\beta\sum \Big(\tilde{x}_{j,1}\tilde{x}_{j-1,1}+\tilde{x}_{j,2}\tilde{x}_{j-1,2}-a_+^2\tilde{b}_{j,1}\tilde{b}_{j-1,1}-a_-^2\tilde{b}_{j,2}\tilde{b}_{j-1,2}\Big)\\ \notag
&+2 \beta \sum \Big( v_j^2\,(\tilde x_{j,1}-\tilde x_{j,2})(\tilde x_{j-1,1}-\tilde x_{j-1,2})+t_j^2\,(a_+\tilde b_{j,1}-a_-\tilde b_{j,2})(a_+\tilde b_{j-1,1}-a_-\tilde b_{j-1,2})\Big)\\ \notag
&-\sum\Big(\dfrac{4}{c_0^2} (\tilde x_{j,1}\tilde x_{j,2}-\tilde b_{j,1}\tilde b_{j,2})
-2(a_+^{-3}n_{j,12}\tilde x_{j,1}\tilde b_{j,2}+a_-^{-3}n_{j,21}\tilde x_{j,2}\tilde b_{j,1})\Big)\\ \notag
&+\Tr A^{-1}\Big(\begin{matrix} \tilde x_{1,1}&0\\ 0& \tilde x_{1,2}\end{matrix} \Big)A^{-1}\Big(\begin{matrix} \tilde x_{1,1}&0\\ 0& \tilde x_{1,2}\end{matrix} \Big).
\end{align}
Here $\zeta=\{\zeta_j\}_{j=1,..,n}$, $\zeta_j=(\zeta_{j,11},\zeta_{j,22}, a_+\zeta_{j,12},a_-\zeta_{j,21})$ with
\begin{align*}
&\zeta_{j,11}=-\Big(U_j\Big(\begin{matrix}\kappa&-iy_1\\ iy_1&-\kappa\end{matrix}\Big)U_j^*\Big)_{11},\quad
\zeta_{j,22}=-\Big(U_j\Big(\begin{matrix}\kappa&-iy_1\\ iy_1&-\kappa\end{matrix}\Big)U_j^*\Big)_{22},\\ \notag
&\zeta_{j,12}=-\Big(S_j\Big(\begin{matrix}\kappa&-iy_2\\ iy_2&-\kappa\end{matrix}\Big)S_j^{-1}\Big)_{11},\quad
\zeta_{j,21}=-\Big(S_j\Big(\begin{matrix}\kappa&-iy_2\\ iy_2&-\kappa\end{matrix}\Big)S_j^{-1}\Big)_{22}.
\end{align*}
We also denoted
\begin{equation}\label{h}
\begin{array}{ll}
h=\{h_{j,ls}+h^{q}_{j,ls}\}_{j=1,..,n,l,s=1,2},& h^0=\{h^0_{j,ls}\}_{j=1,..,n,l,s=1,2},\\ 
h_{j,11}={2}/{c_0}-\beta c_0 v_j^2-\beta c_0 v_{j+1}^2+{a_-n_{j,12}}/{a_+^2}, \quad &h^0_{j,11}={n_{j,1}}/{a_+},\\ 
h_{j,22}=-{2}/{c_0}+\beta c_0 v_j^2+\beta c_0 v_{j+1}^2+{a_+n_{j,21}}/{a_-^2}, \quad &h^0_{j,22}={n_{j,2}}/{a_-},\\ 
h_{j,12}={2a_+}/{c_0}-2-\beta c_0 a_+ t_j^2-\beta c_0 a_+ t_{j+1}^2-n_{j,21}{a_+}/{a_-},\quad  &h^0_{j12}=-{n_{j,1}},\\ 
h_{j,21}=-{2a_-}/{c_0}-2+\beta c_0 a_- t_j^2+\beta c_0 a_- t_{j+1}^2-n_{j,12}{a_-}/{a_+},  \quad &h^0_{j,21}=-{n_{j,2}},
\end{array}
\end{equation}
with
\begin{align*}
&h^{q}_{j,ls}=0,\quad j\ne 1,\\
&h^{q}_{1,11}=-\dfrac{1}{a_+}+(A^{-1})_{11}+\big(A^{-1}\tilde\rho_1S_1(LQ)S_1^{-1}\tilde\tau_1 A^{-1}\big)_{11},\\
&h^{q}_{1,12}=-\dfrac{1}{a_-}+(A^{-1})_{22}+\big(A^{-1}\tilde\rho_1S_1(LQ)S_1^{-1}\tilde\tau_1 A^{-1}\big)_{22},\\
&h^{q}_{1,21}=-1-a_+\big(S_1(LQ)S_1^{-1}\big)_{11},\\
&h^{q}_{1,22}=-1-a_-\big(S_1(LQ)S_1^{-1}\big)_{22},
\end{align*}
and
\begin{align*}
v_j=|(U_jU_{j-1}^*)_{12}|\quad t_j=|(S_jS_{j-1}^{-1})_{12}|.
\end{align*}
\end{lemma}
\begin{proof} Rewriting the determinant in (\ref{D}) in a standard way, we obtain
\begin{equation}\label{exp_det}
\mdet\,\mathcal{D}= \sum\limits_{\bar\sigma} (-1)^{|\sigma|} \prod\limits_{j=1}^n P_{j,\bar\sigma_j}(\tilde{x}_{j,1},\tilde{x}_{2j}, \tilde{b}_{j,1},\tilde{b}_{j,1}),
\end{equation}
where $\bar\sigma$ is a permutation of $\{(j,ls)\}$, $l,s=1,2$, $j=1,..,n$, $\bar\sigma_j$ is its restriction on   $\{(j,ls)\}_{l,s=1}^2$, $(-1)^{|\sigma|}$ is a sign of
$\sigma$ and $P_{j,\bar\sigma_j}$ is an expansion in $\tilde{x}_{j,1},\tilde{x}_{2j}$, $\tilde{b}_{j,1},\tilde{b}_{j,1}$ of the product of four elements from the rows 
$\{(j,ls)\}_{l,s=1}^2$ taken with respect to $\bar\sigma_j$.
 
 Let us prove that
for each $j=1,..,n$ and any $\bar\sigma$ each term of $P_{j,\bar\sigma_j}(\tilde{x}_{j,1},\tilde{x}_{2j}, \tilde{b}_{j,1},\tilde{b}_{j,1})$ of (\ref{exp_det}) belongs to
one of the three following groups:
\begin{itemize}
\item[i.] has a coefficient $W^{-2}$ or lower;

\item[ii.] has a coefficient $W^{-3/2}$ and at least one of variables $\tilde{x}_{j,1},\tilde{x}_{2j}$, $\tilde{b}_{j,1},\tilde{b}_{j,1}$ of the odd degree;

\item[iii.] has a coefficient $W^{-1}$ and at least two variables of $\tilde{x}_{j,1},\tilde{x}_{2j}$, $\tilde{b}_{j,1},\tilde{b}_{j,1}$ of the odd degree;
\end{itemize}

Note that each element in the expansion of the coefficients of the rows $(j,11)$ and $(j,22)$ has a coefficient $W^{-1/2}$ or lower,
and so $P_{j,\bar\sigma_j}(\tilde{x}_{j,1},\tilde{x}_{2j}, \tilde{b}_{j,1},\tilde{b}_{j,1})$ has a coefficient $W^{-1}$ or lower. 
In addition, if $P_{j,\bar\sigma_j}(\tilde{x}_{j,1},\tilde{x}_{j,2}, \tilde{b}_{j,1},\tilde{b}_{j,1})$ contains any terms with $R_{jk}$ (see (\ref{D})),
or at least one off-diagonal elements in $(j,12)$ and $(j,21)$, we get a coefficient  $W^{-2}$ or lower (and so obtain the group (i)).

We are left to consider terms with $d_{j,12}d_{j,21}$. Consider first $j>1$. If   
$P_{j,\bar\sigma_j}(\tilde{x}_{j,1},\tilde{x}_{j,2}, \tilde{b}_{j,1},\tilde{b}_{j,1})$ contains two off-diagonal elements in rows $(j,11)$ and $(j,22)$,
we get group (i). One off-diagonal element and $d_{j,11}$ (or $d_{j,22}$) gives group (ii) or group (i) (since off-diagonal elements do not
depend on $\tilde{x}_{j,1},\tilde{x}_{j,2}$, $\tilde{b}_{j,1},\tilde{b}_{j,1}$), and it is easy to see from (\ref{d_exp1}) that all the terms in expansion of
$d_{j,11}d_{j,22}d_{j,12}d_{j,21}$ belongs to groups (i) -- (iii). For $j=1$ everything will be similar since the zero order term of $T$ of (\ref{T}) (which gives contribution
to the $W^{-1}$ order of elements) does not depend on $\tilde{x}_{j,1},\tilde{x}_{j,2}$, $\tilde{b}_{j,1},\tilde{b}_{j,1}$, and the next orders contribute to
the orders $W^{-3/2}$ or smaller.

To get a non-zero contribution, we have to complete  the expression $P_{j,\bar\sigma_j}(\tilde{x}_{j,1},\tilde{x}_{j,2}, \tilde{b}_{j,1},\tilde{b}_{j,1})$ by some other terms of 
the expansion of the exponent of (\ref{G_last}) in order to get an even degree of each variable $\tilde{x}_{j,1},\tilde{x}_{j,2}$, $\tilde{b}_{j,1},\tilde{b}_{j,1}$. 
But all such a terms have the coefficient $W^{-1/2}$ or lower, and therefore Lemma \ref{l:det_exp} yields that the coefficient near each $j$ in terms that
gives a non-zero contribution must be $W^{-2}$ or lower. Since we have a coefficient $W^{2n}$ in (\ref{G_last}) after the change (\ref{change}), 
this means that to get a non-zero
contribution each coefficient must be exactly $W^{-2}$. Note that the terms of $P_{j,\bar\sigma_j}(\tilde{x}_{j,1},\tilde{x}_{j,2}, \tilde{b}_{j,1},\tilde{b}_{j,1})$ 
that can be completed to the monomial with all even degrees and with a coefficients $W^{-2}$ does not contain any terms with $R_{jk}$,
any terms of (\ref{T}) higher than linear in $\tilde x$'s,
and any terms of the expansion $d_{j,ls}$, $l,s=1,2$ of order $W^{-3/2}$ or lower (except those that comes from $T$). They also cannot be completed to the monomial with all even degrees 
and with a coefficients $W^{-2}$ by any terms of the exponent of (\ref{G_last}) that has a coefficient lower then $W^{-1/2}$ for some $j$.
Thus we need to consider the terms up to the third order in the expansions (\ref{f_exp}) and (\ref{f_exp2}), the linear terms of the functions in the 
exponents (\ref{exp1}) -- (\ref{exp3}), the linear terms coming from
\begin{align}\label{contr1}
 &b_{j,1}^{-2}b_{j,2}^{-2}=\exp\Big\{-
\frac{2\tilde{b}_{j,1}}{\sqrt{W}}-\frac{2\tilde{b}_{j,2}}{\sqrt{W}}+O(W^{-1})\Big\},\\ \notag
&(x_{1,1}x_{1,2}b_{1,1}b_{1,2})^{-1}
= \exp\Big\{-\frac{\tilde{x}_{1,1}}{a_+\sqrt{W}}-\frac{\tilde{x}_{1,2}}{a_-\sqrt{W}}-
\frac{\tilde{b}_{1,1}}{\sqrt{W}}-\frac{\tilde{b}_{1,2}}{\sqrt{W}}+O(W^{-1})\Big\},
\end{align}
and no more than quadratic terms in
\begin{equation}\label{contr2}
\mdet (\hat X_1+i\gamma U_1\sigma U_1^*)=\mdet A\cdot
\exp\Big\{\dfrac{1}{\sqrt{W}}\Tr A^{-1}
\widetilde X_{1}-\frac{1}{2W}\Tr A^{-1}
\widetilde X_{1} A^{-1}
\widetilde X_{1}+O(W^{-3/2})\Big\}
\end{equation}
with
\[
\widetilde X_{1}=\left(\begin{matrix}
\tilde{x}_{1,1}&0\\
0&\tilde{x}_{1,2}
\end{matrix}\right).
\]
Note that the terms containing $\tilde x_{j,1}\tilde b_{j,1}/W$ in $d_{j,11}$ (see (\ref{d_exp1})) cannot contribute to the limit, since if we complete them
to the monomial with even degrees of $\tilde x_{j,1},\tilde b_{j,1}$, then it will contain $W^{-2}$ and an additional $W^{-1}$ should come from the line
containing $d_{j,22}$. Moreover, the terms containing $\tilde x_{j,1}^2$ in $d_{j,11}$ can give a non-zero contribution only if the resulting monomial
contains only $\tilde x_{j,1}^2$, since otherwise, taking into account the contribution of the line
containing $d_{j,22}$, we again obtain at least $W^{-3}$. Thus we can replace $\tilde x_{j,1}^2$ by its average via Gaussian measure 
$(2\pi/c_+)^{-1/2} e^{-c_+\tilde x_{j,1}^2/2}$, i.e. by $ c_+^{-1}$.
The same is true for $\tilde x_{j,2}\tilde b_{j,2}/W$ and for $\tilde x_{j,2}^2$ which could be replaced by $ c_-^{-1}$. Similar argument yields that
the contribution of the terms with $\tilde x_{j,1}^2$ in the line containing $d_{j,12}$ and  $\tilde x_{j,2}^2$ in the line containing $d_{j,21}$  disappear in the limit
$W\to\infty$. Thus the term corresponding to
$W^{2n}\det\mathcal{D}$ in (\ref{G_last})
can be replaced by the term
\begin{eqnarray}\label{det'}
&\int d\rho\,d\tau \exp\Big\{\beta\sum\Tr\Big(U_j^*\tilde\rho_jS_j-U_{j-1}^*\tilde\rho_{j-1}S_{j-1}\Big) \Big(S_j^{-1}\tilde\tau_jU_j-
S_{j-1}^{-1}\tilde\tau_{j-1}U_{j-1}\Big)\\ \notag
&+\sum\Big(c_+n_{j,12}+c_-n_{j,21}-n_{j,1}/c_0a_++ n_{j,2}/c_0a_- \Big)\\ \notag
&+{W}^{1/2}\sum\Big(\big(\tilde{x}_{j,1}/a_+-\tilde{b}_{j,1})n_{j,1}+(\tilde{x}_{j,2}/a_--\tilde{b}_{j,2})n_{j,2}\big)\\
&  -W^{-1/2}\sum\Big(a_+^{-2}\big(\tilde{x}_{j,1}/a_+-\tilde{b}_{j,2}\big)n_{j,12}+a_-^{-2}\big(\tilde{x}_{j,2}/a_--\tilde{b}_{j,1}\big)n_{j,21}\Big)\Big\}
\notag\\ \notag
&\times\exp\Big\{-\Tr \big(A^{-1}-W^{-1/2}A^{-1}
\left(\begin{matrix}
\tilde{x}_{1,1}&0\\
0&\tilde{x}_{1,2}
\end{matrix}\right)
A^{-1}\big)\tilde\rho_1\big(S_1(LQ)S_1^{-1}\big)\tilde \tau_1\Big\}+O(W^{-1/2}),
\end{eqnarray}
where $\tilde \rho_j$, $\tilde \tau_j$, $n_{j,12}$, $n_{j,21}$, $n_{j,1}$, $n_{j,2}$ are defined in (\ref{rt_tilde}).
Here we have used Grassmann variables $\{\rho_{j,ls}\}$, $\{\tau_{j,ls}\}$, $j=1,..,n$, $l,s=1,2$ to rewrite the determinant (\ref{D}) with respect to (\ref{G_Gr}),
have substituted (\ref{d_exp1}) and left only terms that give the contribution (according to arguments above), and then have changed $\rho_{j,11}\to \sqrt{W}\rho_{j,11}$,  
$\tau_{j,11}\to \sqrt{W}\rho_{j,11}$.
Note also 
\begin{equation}\label{capm}
c_+a_+^2=c_0a_+,\quad c_-a_-^2=-c_0a_-.
\end{equation}

Now let us prove that the contribution of the third order in the expansions (\ref{f_exp}) and (\ref{f_exp2}) is small. 
Indeed, the terms  $P_{j,\bar\sigma_j}(\tilde{x}_{j,1},\tilde{x}_{j,2}, \tilde{b}_{j,1},\tilde{b}_{j,1})$ 
that can be completed to the monomial with all even degrees and with a coefficients $W^{-2}$ by these cubic terms cannot
come from the contribution of $T$ of (\ref{T}) and can be one of two types

1. terms $\big(\tilde x_{j,1}/a_+-\tilde{b}_{j,1}\big)\cdot x\cdot c_+\cdot c_-$, where $c_+$, $c_-$ come from the zero terms of $d_{j,12}$, $d_{j,21}$ 
(see (\ref{d_exp1})) and $x$ is an element of the row $(j,22)$ and so
does not depend on $\tilde{x}_{j,1}$, $\tilde{b}_{j,1}$ (or similar terms with $\big(\tilde x_{j,2}/a_--\tilde{b}_{j,2}\big)$);

2. terms of $\big(\tilde x_{j,1}/a_+-\tilde{b}_{j,1}\big)\big(\tilde x_{j,2}/a_--\tilde{b}_{j,2}\big) \big(\tilde x_{j,1}/a_+-\tilde{b}_{j,2}\big) \cdot c_-$
with $\tilde x_{j,1}^2$ or $\tilde b_{j,2}^2$ (or similar terms with $c_+$ coming  from $d_{j,12}$) ;

But it is easy to see that
\[
\int \big(\tilde x^4_{j,1}/(3a_+^4)-\tilde{b}_{j,1}^4/3\big) \cdot e^{-\frac{c_+\tilde x^2_{j,1}}{2}-\frac{a_+^2c_+\tilde b^2_{j,1}}{2}}\,d\tilde x_{j,1}\, d\tilde b_{j,1}= \dfrac{2\pi}{a_+c_+}\Big(\dfrac{1}{a_+^4c_+^2}-\dfrac{1}{a_+^4c_+^2}\Big)=0,
\]
and so the contribution of (1) is zero. Similarly the contribution (2) is zero.

Therefore, the contribution of the third order in the expansions (\ref{f_exp})  is small,  and  
using (\ref{det'}), (\ref{contr1}) -- (\ref{contr2}), and also
\begin{align*}
&\exp\Big\{-\dfrac{1}{n}\sum\limits_{j=1}^n\Big(\Tr U_j^*L_{\pm}U_j\Big(\begin{matrix}\kappa+ix_1&-iy_1\\ iy_1&-\kappa+ix_1\end{matrix}\Big)+\Tr S_j^{-1}L_{\pm}S_j
\Big(\begin{matrix}\kappa+ix_2&-iy_2\\ iy_2&-\kappa+ix_2\end{matrix}\Big)\Big)\Big\}\\
=&\exp\big\{-E(x_1+x_2)\big\}\cdot \exp\Big\{-\dfrac{c_0}{2n}\sum\limits_{j=1}^n\Big(\Tr U_j^*L U_j\Big(\begin{matrix}\kappa&-iy_1\\ iy_1&-\kappa\end{matrix}\Big)+
\Tr S_j^{-1} L S_j\Big(\begin{matrix}\kappa&-iy_2\\ iy_2&-\kappa\end{matrix}\Big)\Big)\Big\}
\end{align*}
for $L_\pm$, $L$ defined in (\ref{L_pm}), we get (\ref{G_main}). 

\end{proof}
Denoting the exponent in the second line of (\ref{G_main}) by $\mathcal{E}(z)$ and taking the Gaussian integral over $d z$ with $z$ of (\ref{rt_tilde}), we get
\begin{align}\label{G_main1}
\int_{\mathbb{R}^{4n}} \mathcal{E}(z)dz&=(2\pi)^{2n}\mdet^{-1/2} M\\
&\exp\Big\{\frac{1}{2}(M^{-1}(W^{1/2}h^0+W^{-1/2}(h+\zeta/\Lambda)), W^{1/2}h^0
+W^{-1/2}(h+\zeta/n))\Big\} .
\notag\end{align}
It is easy to see from (\ref{M}) -- (\ref{tilde_M}) that
\begin{align*}
\det \,M&= \det\,M_0(1+O(W^{-1}))=(c_+^2c_-^2a_+^2a_-^2)^{n} (1+O(W^{-1}))=c_0^{4n}(1+O(W^{-1}))
\end{align*}
with $c_\pm$ of (\ref{c_pm}). Note now that
\[
M^{-1}=\big(M_0+\dfrac{1}{W}\tilde{M}\big)^{-1}=M_0^{-1}-\dfrac{1}{W}M_0^{-1}\tilde{M}M_0^{-1}+O(W^{-2}).
\]
Since $M_0$ is diagonal and $h^0_{j,ls}$ is proportional to $n_{j,1}$ or $n_{j,2}$ and $n_{j,l}^2=0$, we have
\[(M_0^{-1}h^0,h^0)=0.\]
Hence, the exponent in the r.h.s. of (\ref{G_main1}) takes the form
\begin{multline*}
\dfrac{1}{2}\Big((M_0^{-1}h^0,h+\zeta/\Lambda)+(M_0^{-1}(h+\zeta/\Lambda),h^0)\\-(M_0^{-1}\tilde M M_0^{-1}h^0,h^0)\Big)+o(1)
=
I_1+I_2-I_3+o(1).
\end{multline*}
Then we can rewrite (recall (\ref{h}) and (\ref{capm}))
\begin{align}\label{M_0h,h}
I_1&+I_2=\sum \Big(\dfrac{(h_{j,11}+\zeta_{j,11}/n )n_{j,1}}{a_+c_+}+\dfrac{(h_{j,22}+\zeta_{j,22}/\Lambda )n_{j,2}}{a_-c_-}\\ \notag
&-\dfrac{(h_{j,12}+a_+\zeta_{j,12}/n )n_{j,1}}{a_+^2c_+}
-\dfrac{(h_{j,21}+a_-\zeta_{j,21}/n )n_{j,2}}{a_-^2c_-}\Big)\\ \notag
&+\Big(\dfrac{h^q_{1,11}n_{j11}}{a_+c_+}+\dfrac{h_{1,22}n_{1,2}}{a_-c_-}-\dfrac{h_{1,12}n_{1,1}}{a_+^2c_+}
-\dfrac{h_{1,21}n_{1,2}}{a_-^2c_-}\Big)
\end{align}
\begin{align}
\notag &=\sum n_{j,1}\Big(\dfrac{2}{a_+c_0}+\beta\big(t_j^2+t_{j+1}^2-v_j^2-v_{j+1}^2\big)+\dfrac{a_-n_{j,12}}{a_+^2c_0}+\dfrac{n_{j,21}}{a_-c_0}
+\dfrac{\zeta_{j,11}-\zeta_{j,12}}{c_0n}\Big)\\ \notag
&+\sum n_{j,2}\Big(-\dfrac{2}{a_-c_0}+\beta\big(t_j^2+t_{j+1}^2-v_j^2-v_{j+1}^2\big)-\dfrac{a_+n_{j,21}}{a_-^2c_0}-\dfrac{n_{j,12}}{a_+c_0}
-\dfrac{\zeta_{j,22}-\zeta_{j,21}}{c_0n}\Big)\\ \notag
& +n_{1,1}\cdot \dfrac{(A^{-1})_{11}-(S_1(LQ)S_1^{-1})_{11}}{c_+a_+}+n_{1,2}\cdot \dfrac{(A^{-1})_{22}-(S_1(LQ)S_1^{-1})_{22}}{c_-a_-}\\ \notag
&+n_{1,1}n_{1,2}\Big(\dfrac{(A^{-1})_{12}(S_1LQS_1^{-1})_{22}(A^{-1})_{21}}{c_+a_+}+\dfrac{(A^{-1})_{21}(S_1LQS_1^{-1})_{11}(A^{-1})_{12}}{c_-a_-}\Big)+O(W^{-1/2});
\end{align}
\begin{align}\notag
&I_3=\dfrac{4}{c_0^4}\sum n_{j,1}n_{j,2}-\dfrac{1}{a_+^2c_0^2}\sum n_{j,12}n_{j,1}n_{j,2}
-\dfrac{1}{a_-^2c_0^2}\sum n_{j,21}n_{j,1}n_{j,2}\\
 \label{tilMh,h} &+\sum \dfrac{\beta(v_j^2+t_j^2)}{c_0^2}\big(n_{j,1}n_{j+1,1}+n_{j,1}n_{j+1,2}
+n_{j,2}n_{j+1,1}+n_{j,2}n_{j+1,2}\big)\\ \notag
&-\dfrac{1}{c_0^2}(A^{-1})_{12}(A^{-1})_{21} n_{1,1}n_{1,2}+O(W^{-1}).
\end{align}
Moreover,
\begin{align}\label{gr_lapl}
&\exp\Big\{\beta\sum\Tr\Big(U_j^*\tilde\rho_jS_j-U_{j-1}^*\tilde\rho_{j-1}S_{j-1}\Big) \Big(S_j^{-1}\tilde\tau_jU_j-
S_{j-1}^{-1}\tilde\tau_{j-1}U_{j-1}\Big)\Big\}\\ \notag
&=\exp\Big\{\dfrac{\beta}{W}\sum\Tr\Big(U_j^*\hat\rho_jS_j-U_{j-1}^*\hat\rho_{j-1}S_{j-1}\Big) \Big(S_j^{-1}\hat\tau_jU_j-
S_{j-1}^{-1}\hat\tau_{j-1}U_{j-1}\Big)\Big\}+O(W^{-1/2}),
\end{align}
where
\begin{align}\label{rt_hat}
\hat\rho_j=\hbox{diag}\{\rho_{j,11},\rho_{j,22}\},\quad
 \hat\tau_j=\hbox{diag}\{\tau_{j,11},\tau_{j,22}\}.
\end{align}
Combining (\ref{M_0h,h}) -- (\ref{gr_lapl}) we can integrate the main term of (\ref{G_main1}) with respect to $\rho_{j,12}$, $\tau_{j,12}$, $\rho_{j,21}$, $\tau_{j,21}$
according to (\ref{G_Gr}).
This integration gives
\begin{align*}
&\prod\limits_{j=1}^n \Big(c_++\dfrac{a_-n_{j,1}}{a_+^2c_0}-\dfrac{n_{j,2}}{a_+c_0}+\dfrac{n_{j,1}n_{j,2}}{a_+^2c_0^2}\Big)
\Big(c_-+\dfrac{n_{j,1}}{a_-c_0}-\dfrac{a_+n_{j,2}}{a_-^2c_0}+\dfrac{n_{j,1}n_{j,2}}{a_-^2c_0^2}\Big)\\ \notag
&=c_0^2+\dfrac{c_0n_{j,2}}{a_-}-\dfrac{c_0n_{j,1}}{a_+} +\big(1+2/c_0^2\big)n_{j,1}n_{j,2}=
c_0^2\cdot \exp\Big\{-\dfrac{n_{j,1}}{a_+c_0}+\dfrac{n_{j,2}}{a_-c_0}\Big\}\cdot \Big(1+\dfrac{2}{c_0^4}n_{j,1}n_{j,2}\Big),
\end{align*}
which, together with (\ref{M_0h,h}) -- (\ref{gr_lapl}), yields
\begin{align*}
&Z_{\beta nW}(\kappa,z_1,z_2)=c_0^{2n}\cdot e^{E(x_1-x_2)}\int d\hat\rho\, d\hat \tau \,dU\, d S \prod_{j=1}^n
\Big(1-\dfrac{2}{c_0^4}n_{j,1}n_{j,2}\Big)\exp\Big\{-\beta c_0^2\sum (v_j^2+t_j^2)\Big\} \\ \notag
&\times\exp\Big\{\beta\sum\Tr\Big(U_j^*\hat\rho_jS_j-U_{j-1}^*\hat\rho_{j-1}S_{j-1}\Big) \Big(S_j^{-1}\hat\tau_jU_j-
S_{j-1}^{-1}\hat\tau_{j-1}U_{j-1}\Big)\Big\}\\ \notag
&\times \exp\Big\{\sum n_{j,1}\Big(\beta\big(t_j^2+t_{j+1}^2-v_j^2-v_{j+1}^2\big)+\dfrac{\zeta_{j,11}-\zeta_{j,12}}{c_0n}\big)\Big)\Big\}\\
\notag&\times \exp\Big\{\sum n_{j,2}\Big(\beta\big(t_j^2+t_{j+1}^2-v_j^2-v_{j+1}^2\big)-\dfrac{\zeta_{j,22}-\zeta_{j,21}}{c_0n}\big)\Big)\Big\}\\ \notag
&\times \int F(A,Q, \hat\rho_1,\hat\tau_1, S_1) \,dQ\\
&\times\exp\Big\{-\dfrac{c_0}{2n}\sum\limits_{j=1}^n\Big(\Tr U_j^*L U_j\Big(\begin{matrix}\kappa&-iy_1\\ iy_1&-\kappa\end{matrix}\Big)+
\Tr S_j^{-1} L S_j\Big(\begin{matrix}\kappa&-iy_2\\ iy_2&-\kappa\end{matrix}\Big)\Big)\Big\}
+o(1)
\end{align*}
where we have used $a_+c_+=c_0$, $a_-c_-=-c_0$, and
\[
(1+2n_{j,1}n_{j,2}/c_0^4)\cdot e^{-4n_{j,1}n_{j,2}/c_0^4}=1-2n_{j,1}n_{j,2}/c_0^4.
\]
Here
\begin{align*}
&F(A,Q, \hat\rho_1,\hat\tau_1, S_1)=\exp\Big\{-\Tr(iE/2+i\gamma\sigma)(LQ)-\dfrac{c_0}{2}\Tr S_1^{-1}LS_1 (LQ)\Big\}
\\ \notag 
&\times \exp\Big\{-\Tr A^{-1}\hat\rho_1S_1(LQ)S_1^{-1}\hat\tau_1+\dfrac{1}{c_0^2}(A^{-1})_{12}(A^{-1})_{21} n_{1,1}n_{1,2}\Big\}\cdot \mdet \,A
\end{align*}
\begin{align*}
&\times\exp\Big\{n_{1,1}\cdot \dfrac{(A^{-1})_{11}-(S_1(LQ)S_1^{-1})_{11}}{c_0}-n_{1,2}\cdot \dfrac{(A^{-1})_{22}-(S_1(LQ)S_1^{-1})_{22}}{c_0}\Big\}\\ \notag
&\times\exp\Big\{n_{1,1}n_{1,2}\Big(\dfrac{(A^{-1})_{12}(S_1LQS_1^{-1})_{22}(A^{-1})_{21}}{c_0}-\dfrac{(A^{-1})_{21}(S_1LQS_1^{-1})_{11}(A^{-1})_{12}}{c_0}\Big)\Big\}.
\end{align*}
Notice
\begin{align*}
&\exp\Big\{\dfrac{1}{c_0}\big((A^{-1})_{11}n_{1,1}-(A^{-1})_{22}n_{1,2}\big)+\dfrac{1}{c_0^2}(A^{-1})_{12}(A^{-1})_{21} n_{1,1}n_{1,2}\Big\}\cdot \mdet \,A\\
&=\mdet \big(A+\frac{1}{c_0}L\hat\rho_1\hat\tau_1\big),
\end{align*}
where $\hat\rho_1, \hat\tau_1$ is defined in (\ref{rt_hat}).
In addition,
\begin{align*}
&\exp\Big\{-\Tr A^{-1}\hat\rho_1S_1(LQ)S_1^{-1}\hat\tau_1-\dfrac{1}{c_0}(S_1(LQ)S_1^{-1})_{11}n_{1,1}+\dfrac{1}{c_0}(S_1(LQ)S_1^{-1})_{22}n_{1,2}\Big\}\\
&\times\exp\Big\{n_{1,1}n_{1,2}\Big(\dfrac{(A^{-1})_{12}(S_1LQS_1^{-1})_{22}(A^{-1})_{21}}{c_0}-\dfrac{(A^{-1})_{21}(S_1LQS_1^{-1})_{11}(A^{-1})_{12}}{c_0}\Big)\Big\}\\
&=\exp\Big\{\Tr S_1^{-1}\hat\tau_1\Big(A+\dfrac{1}{c_0}L\hat\rho_1\hat\tau_1\Big)^{-1}\hat\rho_1S_1L Q-\dfrac{1}{c_0}\Tr S_1^{-1}L \hat\rho_1\hat\tau_1 S_1(LQ)\Big\}
\end{align*}
hence we can perform the integration with respect to $Q$ to get
\begin{multline*}
\int F(A,Q, \hat\rho_1,\hat\tau_1, S_1) \,dQ\\=\dfrac{\mdet \big(A+\dfrac{1}{c_0}L\hat\rho_1\hat\tau_1\big)}
{\mdet\Big(iE/2+i\gamma\sigma+\dfrac{c_0}{2}S_1^{-1}L\big(1-\dfrac{2}{c_0^2} \hat\rho_1\hat\tau_1\big)S_1-S_1^{-1}\hat\tau_1\big(A+\dfrac{1}{c_0}L\hat\rho_1\hat\tau_1\big)^{-1}\hat\rho_1S_1\Big)}.
\end{multline*}
Using
\[
U_1^*\Big(A+\dfrac{1}{c_0}L\hat\rho_1\hat\tau_1\Big) U_1=
-\dfrac{iE}{2}+\dfrac{c_0}{2}U_1^*L\big(1+\frac{2}{c_0^2} \hat\rho_1\hat\tau_1\big)U_1+i\gamma\sigma,
\]
we get finally
\begin{multline*}
\int F(A,Q, \hat\rho_1,\hat\tau_1, S_1) \,dQ\\=\hbox{sdet}^{-1}
\left(\begin{matrix}
U_1^*L\big(1+\dfrac{2}{c_0^2}\hat\rho_1\hat\tau_1\big)U_1-\dfrac{iE}{c_0}+\dfrac{2i\gamma}{c_0}\sigma& \dfrac{2}{c_0}S_1^{-1}\hat\tau_1 U_1\\
\dfrac{2}{c_0}U_1^*\hat\rho_1 S_1&-\Big(S_1^{-1}L\big(1-\dfrac{2}{c_0^2} \hat\rho_1\hat\tau_1\big)S_1+\dfrac{iE}{c_0}+\dfrac{2i\gamma}{c_0}\Big)
\end{matrix}\right)
\end{multline*}
Now changing
\[
\rho_{j,11}\to c_0 \rho_{j,1}, \quad \tau_{j,11}\to c_0 \tau_{j,1},\quad \rho_{j,22}\to c_0 \rho_{j,2},\quad \tau_{j,22}\to c_0 \rho_{j,2}
\]
with an appropriate change in $n_{j,1}$, $n_{j,2}$, $\hat\rho_j$, $\hat\tau_j$, and recalling (\ref{beta_til}), we get
 (\ref{sigma-mod}) which finishes the proof of Theorem \ref{thm:sigma_mod} for $M=1$. The general case  can be obtain very similar:
 since $M$ is finite, the additional terms (\ref{sup}) do not affect the saddle-points and the main terms in representation (\ref{G_main}), they 
 just add some additional terms to $\tilde M$ of (\ref{tilde_M}), $h^{(q)}_{1,ls}$ of (\ref{h}) and to (\ref{contr1}) -- (\ref{contr2}) which
 can be handled in the same way. 

$\square$

\section{Proof of Theorem \ref{thm:2}}
To simplify formulas below we handle again the case $M=1$. We explain the difference with the case $M>1$ at the end of the section.

It is easy to see that (\ref{sigma-mod}) implies that
$Z_{\beta n}(E,\varepsilon,\xi)$ can be written in the form
\begin{align*}
Z_{\beta n}(E,\varepsilon,\xi)=&e^{E(x_1-x_2)}
\int  D(Q) \tilde{\mathcal{F}}(Q)\tilde{\mathcal{M}}^{n-1}(Q,Q')\tilde{\mathcal{F}}(Q')  dQ dQ'
\end{align*}
where 
\begin{align*}
\tilde{\mathcal{F}}(Q):&=\exp\{ \frac{c_0}{4n}\Str Q\Lambda_{\kappa, y}\}\\
\tilde{\mathcal{M}}(Q,Q')&=\tilde{\mathcal{F}}(Q)\exp\big\{-\dfrac{\tilde\beta}{4}\Str QQ'\big\}\tilde{\mathcal{F}}(Q')
\notag\end{align*}
and
\begin{align*}
&D(Q):=\Sdet^{-1}\Big(Q-\dfrac{iE}{2\pi\rho(E)}+\dfrac{i\gamma}{\pi\rho(E)}\mathcal{L}\Sigma\Big) =D(U,S,\hat\rho,\hat\tau)
\end{align*}

But for the proof of Theorem \ref{thm:2}, it is convenient to  change  variables $\{U_i\}_{i=1}^n$ and $\{S_i\}_{i=1}^n$  in order  to obtain 
a little bit different representation.
\begin{proposition}\label{p:turn}
\begin{align}\label{trans1}
Z_{\beta n}(E,\varepsilon,\xi)=&e^{E(x_1-x_2)}
\int D_1(Q) \mathcal{F}(Q){\mathcal{M}}^{n-1}(Q,Q')\mathcal{F}(Q')  dQ dQ',
\end{align}
where
\begin{align}
\mathcal{F}(Q):&=
\exp\Big\{ \frac{c_0}{4n}\Str Q\Lambda_{1}\Big\},\quad \Lambda_1=\left(\begin{array}{cc}L\kappa_1&0\\0& L\kappa_2 \end{array}\right),
\,\, \kappa_{1,2}=(\kappa^2+y_{1,2}^2)^{1/2},\notag\\
\mathcal{M}(Q,Q')&=\mathcal{F}(Q)\exp\big\{-\dfrac{\tilde\beta}{4}\Str QQ'\big\}\mathcal{F}(Q'),
\notag\\
\label{D_1}
D_1(Q)&=
c_1+c_2n_{1}+c_3n_{2}+c_4n_{1}n_{2}+d_1\rho_{1}\tau_{2}+d_2\rho_{2}\tau_{1},\quad (n_1=\rho_{1}\tau_{1},\,n_2=\rho_{2}\tau_{2})\\
\notag
c_\nu&=\sum_{k=1}^3c_\nu^{(k)}(\tau-i\sinh t\cdot \cos\alpha_2\cos \theta+\cosh t\cdot\sin\alpha_2)^{-k},\quad \nu=1,2,3,4,\\
d_\nu&=d_\nu^{(1)}(\tau-i\sinh t\cdot \cos\alpha_2\cos \theta+\cosh t\cdot\sin\alpha_2)^{-1},\quad \nu=1,2,\notag\\
\tau&=(\gamma+\gamma^{-1})/c_0>0,
\label{tau}\end{align}
and $\alpha_1$, $\alpha_2$ are defined as
\begin{equation}
\sin\alpha_{\sigma}=y_\sigma(\kappa^2+y^2_\sigma)^{-1/2},\quad 0< \alpha_\sigma<\pi/2,\quad \sigma=1,2\label{alpha}.
\end{equation}
Here $c_\nu^{(k)}$  and $d_\nu^{(1)}$ are polynomials with respect to entries of $U$, whose coefficients are independent of
$S$ in the case of $c_\nu^{(k)}$, and are bounded functions of $S$ in the case of  $d_\nu^{(1)}$.
Parameters   $t$, $\theta$ here correspond to the following parametrizations of
 $U\in \mathring{U}(2)$ and $S\in \mathring{U}(1,1)$:
\begin{align}\label{U,S}
U=\left(\begin{array}{cc}e^{i\psi/2}\cos\tfrac\varphi 2 &e^{-i\psi/2}\sin\tfrac\varphi 2 \\
-e^{i\psi/2}\sin\tfrac\varphi 2&e^{-i\psi/2}\cos\tfrac\varphi 2
\end{array}\right),\quad S=\left(\begin{array}{cc}e^{i\theta/2}\cosh\tfrac t 2&e^{-i\theta/2}\sinh\tfrac t 2 \\
e^{i\theta/2}\sinh\tfrac t 2&e^{-i\theta/2}\cosh \tfrac t 2
\end{array}\right).
\end{align}

\end{proposition}

\noindent\textit{Proof.}
Let us introduce unitary matrices
\begin{align*}
V_{\sigma}=\left(\begin{array}{cc}\cos(\alpha_{\sigma}/2)&-i\sin(\alpha_{\sigma}/2)\\
-i\sin(\alpha_{\sigma}/2)&\cos(\alpha_{\sigma}/2)\end{array}\right), \quad \sigma=1,2, 
\end{align*}
where $\alpha_{1,2}$ are defined in (\ref{alpha}). It is straightforward to check  that
\[V_{\sigma}\Lambda_{\kappa,y_\sigma }V_{\sigma}^*=\kappa_{\sigma}L, \quad \sigma=1,2.\]
For the unitary group we  can just change the variables
$ U_i\to U_iV_1$, and since the Haar measure is invariant with respect to this change of variables, we obtain the
desired  transformation for the "unitary" part of $\mathcal{M}$. Unfortunately, similar transformation for the hyperbolic group
does not work directly, since the matrix $\tilde S_i=S_iV_2$ is not hyperbolic. But  if we use another parametrization of the Hyperbolic group
\begin{align*}
S(t,s)=\left(\begin{array}{cc}\cosh(t/2)+ise^{t/2}/2&-\sinh(t/2)-ise^{t/2}/2\\
-\sinh(t/2)+ise^{t/2}/2&\cosh(t/2)-ise^{t/2}/2\end{array}\right),
\end{align*}
then it is straightforward to check that
\[  S(t,s)V_2=S(t+i\alpha_2,s) .  \]
On the other hand,  $\mathcal{M}(S_1,S_2)$ depends only on $S_1S_2^{-1}$ and
the entries of $S_1S_2^{-1}$ depend only on $t_1-t_2$
\begin{align*}
(S(t_1,s_1)S^{-1}(t_2,s_2)))_{11}=\cosh((t_1-t_2)/2)+is_1e^{(t_1-t_2)/2}/2-is_2e^{-(t_1-t_2)/2}/2,\\
(S(t_1,s_1)S^{-1}(t_2,s_2)))_{12}=-\sinh((t_1-t_2)/2)-is_1e^{(t_1-t_2)/2}/2+is_2e^{-(t_1-t_2)/2}/2.
\end{align*}
Hence, if we change the integration contour with respect to all $t_j+i\alpha_2\to t_j$, then
\begin{align*}
&\tilde{\mathcal{F}}\to\mathcal{F},\quad \tilde{\mathcal{M}}\to \mathcal{M},\quad
D(U,S)\to D(UV_1^*,SV_2^*).
\end{align*}
Thus we are left to study $D_1=D(UV_1^*,SV_2^*)$. Denote
\[\tilde U=UV_1^*,\quad \tilde S=SV_2^*.
\]
Using  formulas (\ref{sdet_def}) and (\ref{sigma-mod}), we conclude that
\begin{align}\label{sdet1}
D_1:=&\frac{\det \tilde A}{\det \tilde B}\cdot \frac{\det (1+2L\hat n  \tilde A^{-1})}{\det(1-2L\hat n \tilde B^{-1}+4\hat \rho(\tilde A+2L\hat n)^{-1}\hat\tau)}
\notag\\
\tilde A= & - \frac{iE}{c_0}+\frac{2i\gamma}{c_0}\tilde U\hat \sigma \tilde U^{-1}+L,\quad\tilde B=\frac{iE}{c_0}
+\frac{2i\gamma}{c_0}\tilde S\hat \sigma \tilde S^{-1}+L,\\
\hat n&=\mathrm{diag}\{n_{1},n_{2}\}
\notag\end{align}
It is easy to see that
\begin{align}\label{sdet2}
\det \tilde A=&\det (-\frac{iE}{c_0}+\frac{2i\gamma}{c_0}\tilde U\hat\sigma\tilde U^{-1}+L)=-\frac{E^2}{c_0^2}-1-\frac{4\gamma^2}{c_0^2}-\frac{4i\gamma}{c_0}
(\tilde U_{11} \bar{\tilde U}_{12}- \tilde U_{12}\bar{\tilde U}_{11})=\\
=&-\frac{4(1+\gamma^2)}{c_0^2}+\frac{4\gamma}{c_0}\sin\tilde\varphi\sin\tilde\psi=
-\frac{4\gamma}{c_0}(\tau-\sin\varphi\cdot\cos\alpha_1\sin \psi+\sin\alpha_1\cos\varphi)\notag\\
\det \tilde B=&\det (\frac{iE}{c_0}+\frac{2i\gamma}{c_0}\tilde S\hat\sigma \tilde S^{-1}+L)=-\frac{E^2}{c_0^2}-1-\frac{4\gamma^2}{c_0^2}-\frac{4i\gamma}{c_0}
(\tilde S_{11}( \tilde S^{-1})_{21}-\tilde S_{12}(\tilde S^{-1})_{11})=\notag\\
=&
-\frac{4\gamma}{c_0}(\tau-i\sinh t\cdot \cos\alpha_2\cos \theta+\cosh t\cdot\sin\alpha_2),\notag
\end{align}
where  $\tau$ is defined in (\ref{tau})
and we used parametrizations (\ref{U,S}) for $ U$ and $ S$. Here we used also that
\[
\tilde S^{-1}=\left(\begin{matrix}
\tilde S_{22}&-\tilde S_{12}\\
-\tilde S_{21}&\tilde S_{11}
\end{matrix}\right),
\]
and so
\begin{align}\label{sdet2a}
&-\tilde S_{11}( \tilde S^{-1})_{21}+\tilde S_{12}(\tilde S^{-1})_{11}=\tilde S_{11} \tilde S_{21}+\tilde S_{12}\tilde S_{22}=(SV_2^*)_{11}(SV_2^*)_{21}+(SV_2^*)_{12}(SV_2^*)_{22}\\
&=\cos\alpha_2( S_{11} S_{21}+S_{12} S_{22})+i\sin\alpha_2(S_{11}S_{22}+S_{12}S_{21})\notag\\&=
\sinh t\cdot \cos\alpha_2\cos \theta+i\cosh t\cdot\sin\alpha_2.
\notag\end{align}
Similar formulas can be obtained for $\tilde U_{11} \bar{\tilde U}_{12}- \tilde U_{12}\bar{\tilde U}_{11}$.

Since
\begin{align*}
&\hat \rho(\tilde A+2L\hat n)^{-1}\hat\tau=\hat \rho \tilde A^{-1}(1-2L\hat n \tilde A^{-1})\hat\tau,\\
&\hat \rho\tilde A^{-1}L\hat n \tilde A^{-1}\hat\tau\tilde B^{-1}=-n_1n_2\tilde A^{-1}_{12}\tilde A^{-1}_{21}L\tilde B^{-1},
\end{align*}
we have
\begin{align*}
\det (1-2L\hat n \tilde B^{-1}+&4\hat \rho(\tilde A+2L\hat n)^{-1}\hat\tau\tilde B^{-1})=
\det(1-2L\hat n \tilde B^{-1}+4\hat \rho\tilde A^{-1}(1-2L\hat n\tilde A^{-1})\hat\tau\tilde B^{-1})\\
=&\det(1-2L\hat n \tilde B^{-1}+4\hat \rho\tilde A^{-1}\hat\tau\tilde B^{-1})
\det(1+8n_1n_2\tilde A^{-1}_{12}\tilde A^{-1}_{21}L\tilde B^{-1})\\
=&(1+n_1(4\tilde A^{-1}_{11}-2)\tilde B^{-1}_{11}+n_2(4\tilde A^{-1}_{22}+2)\tilde B^{-1}_{22}
+4\tilde A^{-1}_{12}\tilde B^{-1}_{21}\rho_1\tau_2+4\tilde A^{-1}_{21}\tilde B^{-1}_{12}\rho_2\tau_1)\\
\times&\Big(1+n_1n_2\Big(\frac{(4\tilde A^{-1}_{11}-2)(4\tilde A^{-1}_{22}+2)+16\tilde A^{-1}_{12}\tilde A^{-1}_{21}}{\mathrm{det}\tilde B}+
8\tilde A^{-1}_{12}\tilde A^{-1}_{12}\Tr\tilde B^{-1}L\Big)\Big).
\end{align*}
Hence,
\begin{align*}
\mathrm{det}^{-1}& (1-2L\hat n \tilde B^{-1}+4\hat \rho(\tilde A+2L\hat n)^{-1}\hat\tau\tilde B^{-1})=\Big(1-k_bn_1n_2\Big)
\\
&\times\Big(1-n_1(4\tilde A^{-1}_{11}-2)\tilde B^{-1}_{11}-n_2(4\tilde A^{-1}_{22}+2)\tilde B^{-1}_{22}
-4\tilde A^{-1}_{12}\tilde B^{-1}_{21}\rho_1\tau_2-4\tilde A^{-1}_{21}\tilde B^{-1}_{12}\rho_2\tau_1\Big),
\end{align*}
where
\begin{align*}
k_b=&\frac{(4\tilde A^{-1}_{11}-2)(4\tilde A^{-1}_{22}+2)+16\tilde A^{-1}_{12}\tilde A^{-1}_{21}}{\mathrm{det}\tilde B}+
8\tilde A^{-1}_{12}\tilde A^{-1}_{12}\Tr\tilde B^{-1}L\\&-2(4\tilde A^{-1}_{11}-2)(4\tilde A^{-1}_{22}+2)\tilde B^{-1}_{11}\tilde B^{-1}_{22}
+32\tilde A^{-1}_{12}\tilde A^{-1}_{21}\tilde B^{-1}_{12}\tilde B^{-1}_{21}\\
=&-(16\mathrm{det }^{-1}\tilde A+8\Tr \tilde A^{-1}L-4)(\tilde B^{-1}_{11}\tilde B^{-1}_{22}+\tilde B^{-1}_{12}\tilde B^{-1}_{21})+
8\tilde A^{-1}_{12}\tilde A^{-1}_{21}\Tr\tilde B^{-1}L.
\end{align*}
Similarly,
\begin{align*}
\det (1+2L\hat n  \tilde A^{-1})=(1+2n_1\tilde A^{-1}_{11}-2n_2\tilde A^{-1}_{22})(1-4n_1n_2(\det \tilde A)^{-1}).
\end{align*}
Then finally $D_1$ of (\ref{sdet1}) can be rewritten as
\begin{align}\label{sdet3}
D_1=&(\det \tilde B)^{-1}(\det \tilde A-n_1(4\tilde A_{22}\tilde B^{-1}_{11}-2\tilde A_{22}-2\tilde B^{-1}_{11}\det \tilde A)\\&-
n_2(4\tilde B^{-1}_{22}\tilde A_{11}+2\tilde A_{11}+
2\tilde B^{-1}_{22}\det \tilde A)\notag
+4\tilde A_{12}\tilde B^{-1}_{21}\rho_1\tau_2+4\tilde A_{21}\tilde B^{-1}_{12}\rho_2\tau_1-k n_1n_2),\\
k=&-(16-8\Tr \tilde AL-4\det \tilde A)(\tilde B^{-1}_{11}\tilde B^{-1}_{22}+\tilde B^{-1}_{12}\tilde B^{-1}_{21})-8\,\Tr \tilde B^{-1}L
+4\tilde B^{-1}_{11}\tilde A_{11}+4\tilde B^{-1}_{22}\tilde A_{22}+4,
\notag\end{align}
where we  used 
\begin{align}\label{invA}
\tilde A^{-1}= \left(\begin{matrix}
\frac{\tilde A_{22}}{\det \tilde A}&-\frac{\tilde A_{12}}{\det \tilde A}\\
-\frac{\tilde A_{21}}{\det \tilde A}&\frac{\tilde A_{11}}{\det \tilde A}
\end{matrix}\right).
\end{align}
Using (\ref{invA}) for $\tilde B$, and taking into account that (see (\ref{sdet2a}))
\begin{align*}
\tilde B_{jj}=&\frac{iE}{c_0}-(-1)^j+\frac{2i\gamma}{c_0}(\tilde S\hat \sigma \tilde S^{-1})_{jj}=
\frac{iE}{c_0}-(-1)^j+(-1)^j\frac{2i\gamma}{c_0}(\tilde S_{12}\tilde S_{22}+\tilde S_{11}\tilde S_{21})\\
=&\frac{iE}{c_0}-(-1)^j+(-1)^j\frac{2i\gamma}{c_0}(\sinh t\cdot \cos\alpha_2\cos \theta+i\cosh t\cdot\sin\alpha_2),\quad j=1,2,  \\
&\tilde B^{-1}_{11}\tilde B^{-1}_{22}+\tilde B^{-1}_{12}\tilde B^{-1}_{21}=-(\det\tilde B)^{-1}+2(\det\tilde B)^{-2}\tilde B_{11}\tilde B_{22},
\end{align*}
we obtain (\ref{D_1}).

$\square$

For the next step we will use the following notations:
\begin{align}\label{F}
&F(U,S)=\exp\Big\{-\dfrac{c_0}{n}\Big(\kappa_1\Big(\frac{1}{2}-|U_{12}|^2\Big)
+\kappa_2\Big(\frac{1}{2}+ |S_{12}|^2\Big)\Big)\Big\},\\ \notag
&F_1(U,S)=-\dfrac{c_0}{n}\Big(\kappa_1\Big(\frac{1}{2}-|U_{12}|^2\Big)
-\kappa_2\Big(\frac{1}{2}+ |S_{12}|^2\Big)\Big)
\end{align}
with $\kappa_{1,2}$ defined in (\ref{M}).
\begin{proposition}\label{p:repr}  We have
\begin{align}\label{repr1}
&Z_{\beta n}(\kappa,z_1,z_2)=-\frac{e^{E(x_1-x_2)}}{2\pi i}\oint_{\omega_A}z^{n-1}(\widehat G(z)\widehat f,\widehat g)dz,\quad \omega_A=\{z:|z|=1+A/n\},\\ \label{M_hat}
&  \widehat G(z)=(\widehat M-z)^{-1},\quad\widehat M=\widehat F\widehat K\widehat F,\quad \widehat K=\widehat K_0+O(\beta^{-1}),
\end{align}
where operators $\widehat K_0$, $\widehat F$  and the vectors $\widehat f$, $\widehat g$ have  the form
\begin{align}\label{repr2}
&\quad \widehat K_0=\left(\begin{array}{cccc}K_{US}&\widetilde K_1&\widetilde K_2&\widetilde K_3\\
0&K_{US}&0&\widetilde K_2\\0&0&K_{US}&\widetilde K_1\\0&0&0&K_{US}\end{array}\right),\quad 
\widehat F=F\left(\begin{array}{cccc}1&F_1 &F_1&F_1^2\\
0&1&0&F_1\\0&0&1&F_1\\0&0&0&1\end{array}\right)\\ 
&\hat f=\widehat F (e_4-  e_1),\quad\hat g=\widehat F(c_1e_0+c_2e_2+c_3e_3+(c_4-c_1) e_4)+O(\tilde\beta^{-1})
\label{repr3}\end{align}
with  $F$  and $ F_{1}$  being the  operator of multiplication by the functions $F$  and $F_{1}$ 
defined in (\ref{F}) on $L_2(U)\otimes L_2(S)$ , $K_{US}=K_U\otimes K_S$ and
$K_U$ and $K_S$  being the integral operators in $L_2(U)$ and $L_2(S)$  with a ``difference" kernels
\begin{align*}
K_U(U,U')=K_U(U(U')^*)=\tilde{\beta}  e^{-\tilde{\beta} |(U(U')^*)_{12}|^2},\\
K_S(S,S')=K_S(S(S')^{-1})=\tilde{\beta}  e^{-\tilde{\beta} |(S(S')^{-1})_{12}|^2}.
\notag\end{align*}
and $c_\delta$ having the form (\ref{D_1}).
Here $\widetilde K_p$, $p=1, 2, 3$  are normal operators on $L_2(U)\otimes L_2(S)$, they commute with $K_{US}$ and with the Laplace operators $\widetilde\Delta_U,\widetilde\Delta_S$ 
on the corresponding groups
and satisfy the bounds
\begin{align}\label{bound_re}
&|\widetilde K_p|\le C(1-K_{US})\le -C(\widetilde\Delta_U+\widetilde\Delta_S)/\beta,
\end{align}
where   the Laplace operators $\widetilde\Delta_U,\widetilde\Delta_S$ for the functions depending only on $|S_{12}|^2$ and
$|U_{12}|^2$ have the form
\[\widetilde\Delta_S(\varphi)=-\frac{d}{dx} x(x+1)\frac{d\varphi}{dx} \quad (x=|S_{12}|^2),\qquad 
\widetilde\Delta_U(\varphi)=-\frac{d}{dx} x(1-x)\frac{d\varphi}{dx} \quad (x=|U_{12}|^2).\]
\end{proposition}
The proposition is basically identical to the Proposition 5.1 of \cite{SS:sigma}. The only change is the different form of $\widehat g$ coming from the
presence of the factor $D_1$ in (\ref{trans1}). The from  of $\widehat g$ in (\ref{repr2})
follows from (\ref{sigma-mod}) and Proposition \ref{p:turn}. Indeed, consider the 
 basis  $e_1=1,\,e_2=n_1,\,e_3=n_2,\, e_4=n_1n_2,\,e_5=\rho_1\tau_2,\,e_6=\rho_2\tau_1$, and let $\mathcal{L}_1=\hbox{span}\{e_1,e_2,e_3,e_4\}$.
 Write the transfer operator matrix $H$ as a block matrix with the first block
 corresponding to  $\mathcal{L}_1$ (see the proof of Proposition 5.1 in \cite{SS:sigma}):
\begin{align}\label{H^ij}
&H= \left(\begin{array}{cc}H^{(11)}&H^{(12)}\\
H^{(21)}&H^{(22)}\end{array}\right),\quad
H^{(22)}=\left(\begin{array}{cc}h_{11}&h_{12}\\
h_{21}&h_{22}\end{array}\right),\\
&H^{(21)}=\left(\begin{array}{cccc}2x_d& x&x&0\\
-2\overline x_d&-\overline x&-\overline x&0\end{array}\right),\quad H^{(12)}=\left(\begin{array}{cc}0&0\\  y& -\overline{ y}\\
 y& -\overline{ y}\\
2 y_d&-2\overline{ y_d}
\end{array}\right).
\notag\end{align}
Here 
 $h_{ij}, x,y,x_d,y_d$ are  ``difference" operators whose kernels are defined with the functions
\begin{align}\label{hxy}
h_{ij}=&h_{ijU}h_{ijS},\quad h_{ijU}=U_{ij}^2K_U,\quad h_{ijS}=\bar S_{ij}^2K_S\\
x=&x_Ux_S,\quad x_U=U_{11}U_{12}K_U,\quad x_S=\bar S_{11}\bar  S_{12}K_{S},\quad x_d=x\cdot d,\notag\\
 y=&y_Uy_S,\quad y_U=U_{11}\bar  U_{12}K_U\quad y_S=\bar S_{11}S_{12}K_{S},\quad
y_d=y\cdot d,
\notag\end{align}
and $\bar x,\bar y,\bar x_d,\bar y_d$ mean the complex conjugate kernels. We recall that we are saying that the operator in
$L_2(\mathring U_2)$
is a ``difference" one with a kernel $f$, if its kernel $k(U_1,U_2)$ has the form $k(U_1,U_2)=f(U_1U_2^*)$. The  operator on $L_2(\mathring U(1,1))$ 
is a``difference" one with a kernel $f$, if  $k(S_1,S_2)=f(S_1S_2^{-1})$. 
Let us recall also that $K$ (and consequently its resolvent) was obtained from $H$ by  the transformation
\[
K=\hat T H\hat T, \quad \hat T=\mathrm{diag}\{T, I\},\quad  T=\left(\begin{array}{cccc}
0&0&0&\tilde{\beta}\\0&0&1&0\\0&1&0&0\\\tilde{\beta}^{-1}&0&0&0
\end{array}\right)
\]
Hence the entries of the off-diagonal blocks of the resolvent of $K$ are obtained from those of the off-diagonal blocks of the resolvent of $H$ by multiplication by $\beta$, $1$, or $\beta^{-1}$. Thus, to obtain the bound 
$O(\beta^{-1})$ in (\ref{repr3}), it is sufficient to get the bound $O(\beta^{-2})$ for the corresponding entries of
 the resolvent of $H$.

According to the Schur formula, the resolvent $G(z)=(H-z)^{-1}$  has the form
\begin{align*}
G(z)&:=
 \left(\begin{array}{cc}G^{(11)}&-G^{(11)}H^{(12)}G_2\\-G_2H^{(21)}G^{(11)}&G_2
 +G_2H^{(21)}G^{(11)}H^{(12)}G_2\end{array}\right),\\
G^{(11)}&=M_1^{-1},\quad M_1=H^{(11)}-z-H^{(12)}G_2H^{(21)},\quad G_2(z)=(H^{(22)}-z)^{-1}.
\notag\end{align*}
Since $\hat f_{5,6}=0$, we can write
\begin{align*}
(G(z)\hat f,\hat g)=(G^{(11)}f^{(1)},g^{(1)}-(H^{(21)})^*G_2^*g^{(2)}),
\end{align*}
where $f^{(1)},g^{(1)}$ are the projection  of $\hat f, \hat g$ on $\mathcal{L}_1$ and  $g^{(2)}$ is a projection of $\hat g$ on $\hbox{span}\{e_5,e_6\}$.
Let us consider $H^{(22)}=\hat h+\tilde h$, where $\hat h$  is a diagonal part and $\tilde h$-off diagonal part of $ H^{(22)}$,
and let $G_{2d}=(\hat h-z)^{-1}$.
By the resolvent identity we can write
\begin{align}\label{exp_G_2}
G_2=G_{2d}-G_{2d}\tilde hG_{2}
\end{align}
Moreover, it  was  proven in \cite{SS:sigma} (see the proof of Lemma 6.2) that
\begin{align*}
&\|G_2(z)\|\le Cn,\quad \|G_{2d}(z)\|\le Cn,\quad \|\tilde h\|\le C\beta^{-2},\quad\|H^{(21)}\|\le C\beta^{-2}\\
&\Rightarrow \|(H^{(21)})^*(G_{2d} \tilde h G_{2})^*\|\le C\beta^{-2}
\end{align*}
For the first  terms of the r.h.s. of (\ref{exp_G_2}) we use the expansion
\begin{align}\label{exp_G_2a}
G_{2d}(z)=-z^{-1}\sum_{s=0}^\infty z^{-s}(\hat h)^s
\end{align}
It is easy to see that, due to the form $g^{(2)}$ (see (\ref{sdet3})), $\hat h$ and $H^{(21)}$ (see (\ref{H^ij}) and (\ref{hxy})), after the integration
with respect to $U$ only the term corresponding to
 $s=1$  in the above expansion will give non zero contribution. Hence, using that  
 \[
 \|(H^{(21)})^*\hat h^*g^{(2)}\|\le C\beta^{-2}, 
 \]
after the multiplication by $\beta$ we get (\ref{repr3}).

$\square$

Now let us  derive (\ref{t2.1}) from Proposition \ref{p:repr}.
To this end, set  
\begin{align*}
\widehat M_0=\widehat F^2,\quad \widehat G_0=(\widehat M_0-z)^{-1},
\end{align*}
and consider
\begin{align*}
\Delta G:=\widehat G-\widehat G_0=-\widehat G_0(\widehat M-\widehat M_0)\widehat G_0-\widehat G_0(\widehat M-\widehat M_0)\widehat G(\widehat M-\widehat M_0)\widehat G_0.
\end{align*}
We  apply the following lemma:
\begin{lemma}\label{l:bG}
For any $z\in\omega_A$ (see (\ref{repr1})) we have the bounds
\begin{align}\label{b_1}
& \|\widehat G\|\le C\log^2n/|z-1|\quad \|(\widehat M-\widehat M_0)\widehat G_0\widehat f\|^2\le C(n/\tilde{\beta})^2,\quad
\\
&\|(\widehat M-\widehat M_0)\widehat G_0\widehat g\|^2\le C(n/\tilde{\beta})^2,\quad
 |(\widehat G_0(\widehat M-\widehat M_0)\widehat G_0\widehat f,\widehat g)|\le \dfrac{n\log n}{\tilde{\beta} |z-1|}
\label{b_Gg}\end{align} 
\end{lemma}
Inequalities (\ref{b_1}) were proven in  \cite{SS:sigma} (see Lemma 5.1).  Hence we need to prove only inequalities (\ref{b_Gg}). We postpone the proof to the end of the
section, and continue with the proof of  (\ref{t2.1}) using Lemma \ref{l:bG}. 

\noindent Let us write 
\begin{align*}
&\Big|\frac{1}{2\pi i}\oint_{\omega_A}z^{n-1}(\Delta G\widehat f,\widehat g)dz\Big|\le C\oint_{\omega_A} |(\widehat G_0(\widehat M-\widehat M_0)\widehat G_0\widehat f,\widehat g)|\,|dz|\\
&+ C\oint_{\omega_A}\|\widehat G(z) \|\cdot \|(\widehat M-\widehat M_0)\widehat G_0(z)\widehat f\|\cdot
\|(\widehat M-\widehat M_0)\widehat G_0(\bar z)\widehat g\|\,|dz|\\
&\le C(n\log n/\tilde{\beta})\oint_{\omega_A}\frac{|dz|}{|z-1|}\le Cn\log^2 n/\tilde{\beta}\to 0,
\end{align*}
where we used $n\log^2 n\ll \tilde\beta$ and
\[
\oint_{\omega_A}\frac{|dz|}{|z-1|}\le C\log n.
\]
Thus we have proved that (recall (\ref{repr2}))
\begin{align*}
Z_{\beta n}(\kappa,z_1,z_2)&=-\frac{e^{E(x_1-x_2)}}{2\pi i}\oint_{\omega_A}z^{n-1}
(\widehat G_0(z)\widehat f,\widehat g)dz+o(1)=e^{E(x_1-x_2)}(\widehat F^{2n-2}\widehat f,\widehat g)+o(1).
\end{align*}

 \medskip
 
 \textit{Proof of inequalities (\ref{b_Gg})}

Using  inequalities (\ref{bound_re}) it is easy to conclude the it is sufficient to prove that

\begin{align}
&\|(\Delta_S+\Delta_U)\widehat G_0\widehat g\|^2\le Cn^2,\quad
 |(\widehat G_0(\Delta_S+\Delta_U)\widehat G_0\widehat f,\widehat g)|\le n\log n/|z-1|.
\label{b_Gg.1}
\end{align}
Notice that
\begin{align}\notag
   \widehat G_0=&(\widehat F^2-z)^{-1} =\left(\begin{array}{cccc}G_0&G_0F_1FG_0&G_0F_1FG_0&G_0F_1FG_0F_1FG_0+G_0F_1^2FG_0\\
0&G_0&0&G_0F_1FG_0\\0&0&G_0&G_0F_1FG_0\\0&0&0&G_0\end{array}\right)\\ G_0=&(F^2-z)^{-1},\quad F(x)=e^{-(2c_2x-2c_1u+c_1-c_2)/2n},
\quad (x=|S_{12}|^2),\quad (u=|U_{12}|^2),
\label{G_0,hatG_0}\end{align}
and observe that coefficients of $ \widehat G_0$ do not depend on $\theta$ of (\ref{U,S}). Hence we can integrate over $\theta$ in expression
for $\hat g$ of (\ref{repr2}) and (\ref{D_1}). Using that
\begin{align}
\frac{1}{2\pi}&\int\limits_0^{2\pi} \frac{d\theta}{(\tau-i\cos\alpha_2\sinh t\cos\theta+\sin\alpha_2\cosh t)^\delta}\\
&\qquad=C_\delta\frac{\partial^{\delta-1}}{\partial\tau^{\delta-1}}\Big((\tau+\sin\alpha_2\cosh t)^2
+\sinh^2 t\cos^2\alpha_2\Big)^{-1/2},
\quad\delta=1,2,3\notag
 \end{align}
one can conclude that (\ref{b_Gg.1}) will follow from the bounds
\begin{align}
&\|(\Delta_S+\Delta_U)\tilde f_\nu\tilde g_\nu\|^2\le Cn^2,\quad (\nu=1,\dots,4)
\label{b_Gg.2}
\end{align}
where for $x=\sinh^2 (t/2)$ 
\begin{align}\label{f_alpha}
\tilde  f_\nu(x)=&a_\nu G_0F+b_\nu F_1(FG_0)^2+d_\nu F_1^2(FG_0)^3\\
  \tilde g_\nu(x)=&a_\nu ^{(1)}\Big((\tau+\sin\alpha_2\cosh t)^2
+\sinh^2 t\cos^2\alpha_2\Big)^{-1/2}\notag\\
&+b_\nu^{(1)}\frac{\partial}{\partial\tau}\Big((\tau+\sin\alpha_2\cosh t)^2
+\sinh^2 t\cos^2\alpha_2\Big)^{-1/2}\notag\\
  &+d_\nu^{(1)} \frac{\partial^2}{\partial\tau^2}\Big((\tau+\sin\alpha_2\cosh t)^2
+\sinh^2 t\cos^2\alpha_2\Big)^{-1/2},
\notag\end{align}
where $a_\nu,b_\nu,d_\nu$ and $a^{(1)}_\nu,b^{(1)}_\nu,d^{(1)}_\nu$ are bounded functions depending only on $u$.

It is straightforward to check that
\begin{align}\label{cond_g}
&|\tilde g_\nu(x)|\le C(x^2+1)^{-1/2},\quad |(x+1)\tilde g_\nu^{\prime}(x)|\le C(x^2+1)^{-1/2},\\
&x(x+1)|\tilde g_\nu^{\prime\prime}(x)|\le C(x^2+1)^{-1/2},
\notag\end{align}
and
\[|\tilde f_\nu^{\prime\prime}(x)|\le  Cn,\quad |(x+1)f_\nu^{\prime}(x)|\le Cn,\quad|f_\nu(x)|\le Cn.
\]
Then, since
\[\Delta_S(\tilde f_\nu\tilde g_\nu)=x(x+1)(\tilde f_\nu\tilde g_\nu)^{\prime\prime}+(2x+1)(\tilde f_\nu\tilde g_\nu)^{\prime}
+\tilde f_\nu\tilde g_\nu,\]
we conclude that
\[|\Delta_S(\tilde f_\nu\tilde g_\nu)|\le Cn(x^2+1)^{-1/2}\quad\Rightarrow\quad 
\|\Delta_S(\tilde f_\nu\tilde g_\nu)\|^2\le Cn^2.\]
In addition, one can obtain by the same way that
\[\|\Delta_U(\tilde f_\nu\tilde g_\nu)\|^2\le Cn^2.\]
Thus, we obtain (\ref{b_Gg.2}).

To prove the second inequality in (\ref{b_Gg.1}), we observe that for any $f_\nu$ of the same type as in (\ref{f_alpha}) we have
\begin{align*}
|G_0\Delta_S f_\nu|\le \frac{Ce^{-cx/n}}{|z-1|^2}\le \frac{Cne^{-cx/n}}{|z-1|}.
\end{align*}
Hence
\begin{align*}
\Big |\int \Delta_S( f_\nu)g_\nu dx\Big|&\le \frac{Cn}{|z-1|}\int_0^\infty\frac{e^{-cx/n}dx}{(x^2+1)^{1/2}}\\
&=\frac{Cn}{|z-1|}\int_0^\infty\frac{e^{-c\tilde x}dx}{(\tilde x^2+n^{-2})^{1/2}}\le \frac{Cn\log n}{|z-1|}.
\end{align*}
Here we changed the variable $x\to n\tilde x$. Repeating the argument for $\Delta_U$ we obtain the second inequality in (\ref{b_Gg.2}). 

$\square$

\medskip

The  case $M>1$  is very similar, since in this case the transfer matrix $\mathcal{M}$ and $\mathcal{F}$ in (\ref{trans1}) remain
the same and only $D_1$ is replaced by the product of $D_\alpha$  of the same form but with different $\gamma_\alpha$ (see (\ref{sigma-mod})).
Hence, in (\ref{repr1}) the resolvent $\widehat G$ and the function $\widehat f$ are the same and only the function $\widehat g$ will be different.
But one can see from the argument given after (\ref{cond_g}) that for our proof  we need only bounds (\ref{cond_g}), and the fact that 
$\widehat g$ depends polynomially on entries of $U$ (recall that we used polynomial dependence on $U$ in (\ref{exp_G_2})-(\ref{exp_G_2a}) to prove that
only a finite number of terms in (\ref{exp_G_2a}) are non zero). But  for $M>1$ $D_1$ should
be replaced by the product of $D_\alpha$ and each of them has the form (\ref{D_1}) with $\tau$ replaced by $\tau_\alpha$ defined by (\ref{tau})
with $\gamma=\gamma_\alpha$. Hence  it is evident that that resulting $\widehat g$ will satisfy (\ref{cond_g}) and will depend on entries of $U$
polynomially.

\section{Appendix}

\subsection{Grassmann integration}
Let us consider two sets of formal variables
$\{\psi_j\}_{j=1}^n,\{\overline{\psi}_j\}_{j=1}^n$, which satisfy the anticommutation
conditions
\begin{equation}\label{anticom}
\psi_j\psi_k+\psi_k\psi_j=\overline{\psi}_j\psi_k+\psi_k\overline{\psi}_j=\overline{\psi}_j\overline{\psi}_k+
\overline{\psi}_k\overline{\psi}_j=0,\quad j,k=1,\ldots,n.
\end{equation}
Note that this definition implies $\psi_j^2=\overline{\psi}_j^2=0$.
These two sets of variables $\{\psi_j\}_{j=1}^n$ and $\{\overline{\psi}_j\}_{j=1}^n$ generate the Grassmann
algebra $\mathfrak{A}$. Taking into account that $\psi_j^2=0$, we have that all elements of $\mathfrak{A}$
are polynomials of $\{\psi_j\}_{j=1}^n$ and $\{\overline{\psi}_j\}_{j=1}^n$ of degree at most one
in each variable. We can also define functions of
the Grassmann variables. Let $\chi$ be an element of $\mathfrak{A}$, i.e.
\begin{equation}\label{chi}
\chi=a+\sum\limits_{j=1}^n (a_j\psi_j+ b_j\overline{\psi}_j)+\sum\limits_{j\ne k}
(a_{j,k}\psi_j\psi_k+
b_{j,k}\psi_j\overline{\psi}_k+
c_{j,k}\overline{\psi}_j\overline{\psi}_k)+\ldots.
\end{equation}
For any
sufficiently smooth function $f$ we define by $f(\chi)$ the element of $\mathfrak{A}$ obtained by substituting $\chi-a$
in the Taylor series of $f$ at the point $a$. Since $\chi$ is a polynomial of $\{\psi_j\}_{j=1}^n$,
$\{\overline{\psi}_j\}_{j=1}^n$ of the form (\ref{chi}), according to (\ref{anticom}) there exists such
$l$ that $(\chi-a)^l=0$, and hence the series terminates after a finite number of terms and so $f(\chi)\in \mathfrak{A}$.

Following Berezin \cite{Ber}, we define the operation of
integration with respect to the anticommuting variables in a formal
way:
\begin{equation*}
\intd d\,\psi_j=\intd d\,\overline{\psi}_j=0,\quad \intd
\psi_jd\,\psi_j=\intd \overline{\psi}_jd\,\overline{\psi}_j=1,
\end{equation*}
and then extend the definition to the general element of $\mathfrak{A}$ by
the linearity. A multiple integral is defined to be a repeated
integral. Assume also that the ``differentials'' $d\,\psi_j$ and
$d\,\overline{\psi}_k$ anticommute with each other and with the
variables $\psi_j$ and $\overline{\psi}_k$. Thus, according to the definition, if
$$
f(\psi_1,\ldots,\psi_k)=p_0+\sum\limits_{j_1=1}^k
p_{j_1}\psi_{j_1}+\sum\limits_{j_1<j_2}p_{j_1,j_2}\psi_{j_1}\psi_{j_2}+
\ldots+p_{1,2,\ldots,k}\psi_1\ldots\psi_k,
$$
then
\begin{equation*}
\intd f(\psi_1,\ldots,\psi_k)d\,\psi_k\ldots d\,\psi_1=p_{1,2,\ldots,k}.
\end{equation*}

   Let $A$ be an ordinary Hermitian matrix with a positive real part. The following Gaussian
integral is well-known
\begin{equation}\label{G_C}
\intd \exp\Big\{-\sum\limits_{j,k=1}^nA_{jk}z_j\overline{z}_k\Big\} \prod\limits_{j=1}^n\dfrac{d\,\Re
z_jd\,\Im z_j}{\pi}=\dfrac{1}{\mdet A}.
\end{equation}
One of the important formulas of the Grassmann variables theory is the analog of this formula for the
Grassmann algebra (see \cite{Ber}):
\begin{equation}\label{G_Gr}
\int \exp\Big\{-\sum\limits_{j,k=1}^nA_{jk}\overline{\psi}_j\psi_k\Big\}
\prod\limits_{j=1}^nd\,\overline{\psi}_jd\,\psi_j=\mdet A,
\end{equation}
where $A$ now is any $n\times n$ matrix.

%
%
We will also need the following bosonization formula
\begin{proposition}({\bf see, e.g., \cite{SupB:08} })\label{p:supboz}\\
Let $F:\mathbb{R}\to\mathbb{C}$ be some function that depends only on combinations 
\begin{align*}
\bar{\phi}\phi:=\Big\{\sum\limits_{\alpha=1}^p \bar{\phi}_{l\alpha}\phi_{s\alpha}\Big\}_{l,s=1}^2,
\end{align*}
and set
\[
d\Phi=\prod\limits_{l=1}^2\prod\limits_{\alpha=1}^p d\Re \phi_{l\alpha} d\Im \phi_{l\alpha}.
\]
Assume also that $p\ge 2$. Then
\begin{equation*}
\int F\left(\bar{\phi}\phi \right)d\Phi=\dfrac{\pi^{2p-1}}{(p-1)!(p-2)!}\int F(B)\cdot \mdet^{p-2} B \,dB,
\end{equation*}
where
$B$ is a $2\times 2$ positive Hermitian matrix, and
\begin{align*}
dB&=\mathbf{1}_{B>0}dB_{11}dB_{22}d\Re B_{12} d\Im B_{12}.
\end{align*}
\end{proposition}

\bigskip

{\bf Acknowledgements.} The paper is based upon work supported by the NSF grant DMS-1928930 while the authors participated in a program 
"Universality and Integrability in Random Matrix Theory and Interacting Particle Systems"
hosted by the Mathematical Sciences Research Institute in Berkeley, California, during the Fall 2021 semester.

\end{document}